\documentclass[final,12pt]{colt2023}

\usepackage[colorinlistoftodos, textsize=tiny]{todonotes}
\usepackage{mathtools}
\usepackage{booktabs}

\usepackage{thmtools}
\usepackage{thm-restate}
\usepackage{cleveref}

\usepackage{makecell}
\usepackage{dsfont}
\usepackage{pgf}
\usepackage{siunitx}

\usepackage{pdflscape}

\DeclareRobustCommand{\rchi}{{\mathpalette\irchi\relax}}
\newcommand{\irchi}[2]{\raisebox{\depth}{$#1\chi$}} 

\renewcommand{\d}{\mathrm{d}}
\renewcommand{\Pr}[1]{\mathrm{Pr}\left[#1\right]}
\newcommand{\abs}[1]{\left|#1\right|} 
\newcommand{\norm}[1]{\left\|#1\right\|}
\newcommand{\Expected}[1]{\mathbb{E}\left[#1\right]}
\newcommand{\Prsub}[2]{\text{Pr}_{#1}\left[#2\right]}
\newcommand{\Event}[1]{\mathbf{E}_{#1}}

\newcommand{\F}[1]{\mathbf{#1}}
\newcommand{\vol}{\ensuremath{\nu}}

\newcommand*{\erdosGraph}{Erdős--Rényi graph\xspace}
\newcommand*{\erdosGraphs}{Erdős--Rényi graphs\xspace}

\newcommand{\Gtest}{\tilde{G}}
\newcommand{\teststatistic}[1]{\textsc{CC}^{(+)}(#1)}

\title[Determining the dimensionality of networks]{Real-World Networks are Low-Dimensional: Theoretical and Practical Assessment}
\usepackage{times}

\coltauthor{\Name{Tobias Friedrich} \Email{tobias.friedrich@hpi.de}\\
 \Name{Andreas Göbel} \Email{andreas.goebel@hpi.de}\\
 \Name{Leon Schiller} \Email{leon.schiller@student.hpi.de}\\
 \addr Hasso Plattner Institute, University of Potsdam
 \AND
 \Name{Maximilian Katzmann} \Email{maximilian.katzmann@kit.edu}\\
 \addr Karlsruhe Institute of Technology
}

\begin{document}

\maketitle

\begin{abstract}%
 
Detecting the dimensionality of graphs is a central topic in machine learning.  While the problem has been tackled empirically as well as theoretically, existing methods have several drawbacks.  On the one hand, empirical tools are computationally heavy and lack theoretical foundation.  On the other hand, theoretical approaches do not apply  to graphs with heterogeneous degree distributions, which is often the case for complex real-world networks.

To address these drawbacks, we consider geometric inhomogeneous random graphs (GIRGs) as a random graph model, which captures a variety of properties observed in practice.  These include a heterogeneous degree distribution and non-vanishing clustering coefficient, which is the probability that two random neighbours of a vertex are adjacent.  In GIRGs, $n$ vertices are distributed on a $d$-dimensional torus and weights are assigned to the vertices according to a power-law distribution.  Two vertices are then connected with a probability that depends on their distance and their weights.

Our first result shows that the clustering coefficient of GIRGs scales inverse exponentially with respect to the number of dimensions, when the latter is at most logarithmic in $n$. This gives a first theoretical explanation for the low dimensionality of real-world networks observed by \citet{almagro2022detecting}. A   key element of our proof is to show that when $d=o(\log n)$ the clustering coefficient concentrates around its expectation and that it is dominated by the clustering coefficient of the low-degree vertices.

We further use these insights to derive a linear-time algorithm for determining the dimensionality of a given GIRG. We prove that our algorithm returns the correct number of dimensions with high probability when the input is a GIRG.  As a result, our algorithm bridges the gap between theory and practice, as it not only comes with a rigorous proof of correctness but also yields results comparable to that of prior empirical approaches, as indicated by our experiments on real-world instances.
\end{abstract}


\section{Introduction}
A key technique for understanding and analysing large complex data sets is to embed them into a low-dimensional geometric space. Hence, the search for embedding and dimensionality reduction algorithms has become an important direction in data analysis and machine learning research~\citep{belkin2001laplacian,sarveniazi2014actual,camastra2016intrinsic,Nickel_Kiela_2017}. Embedding algorithms commonly require a metric that captures the similarities between data points, which is often abstracted using a graph whose vertices represent the data points and two vertices are connected if they are close with respect to this metric. The algorithm then determines geometric positions for these vertices such that connected vertices are close together. Such approaches often require an a priori knowledge of the dimensionality, which is unknown in most applications. Heuristic approaches try to determine the dimensionality of a dataset by embedding it in spaces of different dimensionality and choosing the value that yields the optimal embedding~\citep{levina2004maximum,yin2018dimensionality,gu2021principled}.

The recent work of~\citet{almagro2022detecting} gives a new algorithm for learning the dimensionality that does not require embeddings. Instead, given a graph as input, their method counts the number of short -- i.e. length 3, 4 and 5 -- cycles of a graph. It then generates a search space consisting of random graphs that are generated from a geometric model of varying parameters, including the dimensionality of the space. Finally, a data-driven classifier finds the random graph of the search space that resembles the input graph the most and returns its dimensionality. A remarkable observation, that comes from using their algorithm to learn the dimensionality of real-world networks, is that the vast majority of networks has very low dimensionality, which is independent the size of the network.

A downside of the aforementioned approaches is that they rely on machine learning techniques that are computationally heavy and lack theoretical explanation. In order to argue with mathematical rigour, one requires to work with well-defined mathematical objects. A common approach to incorporate such an object is that of average-case analysis, that is, assume that the input graph comes from a well-defined random graph model. The random graph model that has been mostly considered so far in the literature is that of spherical random graphs, where vertices are generated independently and uniformly at random as points on the surface of a $d$-dimensional sphere and two vertices are connected if their angle is bellow a certain threshold. It can be easily shown that, as the number of dimensions increases, spherical random graphs converge to \erdosGraphs, the classical random graph model where edges are drawn independently. A series of works considers the statistical testing problem of detecting weather a given graph is a spherical random graph or an \erdosGraph and determines the parameter regime under which this can be done \citep{Devroye_Gyoergy_Lugosi_Udina_2011,Bubeck_Ding_Eldan_Racz_2015,brennan2020phase,liu2021probabilistic,Liu2022TestingThresholds}. Follow up works consider noisy settings \citep{liu2021phase} or anisotropic geometric random graphs~\citep{eldan2020information,brennan2022threshold}, where each dimension has a different influence on the drawing of edges. The techniques of the aforementioned results can also be used for determining the dimensionality of the given graph \citep[Theorem 5]{Bubeck_Ding_Eldan_Racz_2015}.

A characteristic of the random graph models considered in the aforementioned works, i.e. spherical random graphs and \erdosGraphs, is that the degree distributions of the generated graphs is concentrated around its expected value; this contrasts the power-law degree distributions observed in real-world networks~\citep{faloutsos1999power}. While a latent geometric space appears to be a fundamental requirement for a random graph model that captures the high clustering coefficient~\citep{krioukov2016clustering,boguna2021network} and small diameter \citep{friedrich2013diameter} observed in real-world networks, one needs to also consider the heterogeneity observed in the vertex degrees. A popular model in network theory capturing all previous properties is based on generating points on the hyperbolic plane instead of Euclidean \citep{Boguna_Papadopoulos_Krioukov_2010}.
However, it is not clear what the non-geometric counterpart to test against is in this case.

In this article we bring theory and practice closer together and provide a rigorous explanation for the very low dimensionality of real-world networks that has been observed in practice. Our proofs give new insights with which we are able to design linear-time algorithms for learning the dimensionality of a network and show that they give the correct answer with high probability. To achieve our goal we consider the following random graph models.

\paragraph{Geometric inhomogeneous random graphs \& Chung--Lu graphs.}

Geometric inhomogeneous random graphs (GIRGs), introduced by \citet{Bringmann_Keusch_Lengler_2017} and are defined as follows.

Let $G(n, d, \beta, w_0) = (V, E)$ denote the $n$-vertex graph obtained in the following way. For each $v \in V$, we sample a weight $w_v$ from the Pareto distribution $\mathcal{P}$ with parameters $w_0,1 - \beta$ such that the CDF and density is \begin{align*}
    \Pr{w_v \le x} = 1 - \left(x/w_0\right)^{1-\beta} \text{ and } \rho_{w_v}(x) = \frac{\beta - 1}{w_0^{1-\beta}}x^{-\beta},
\end{align*} respectively. We denote the sequence of the drawn weights by $\{w\}_1^n$ and assume that $\beta > 2$ such that a single weight has finite expectation (and thus the average degree in the graph is constant), but possibly infinite variance. Moreover, each vertex $v$ is assigned a position $\F{x}_v$ in the $d$-dimensional torus $\mathbb{T}^d$ uniformly at random according to the standard Lebesgue measure. We denote the $i$-th component of $\F{x}_v$ by $\F{x}_{v}(i)$. Two vertices $u,v$ are adjacent if and only if their distance $d(\F{x}_u, \F{x}_v)$ is at most the \emph{connection threshold} $t_{uv}$, which is defined such that the marginal connection probability of $u,v$ is \begin{align}\label{eq:conprob}
    \Pr{u \sim v} \coloneqq \min\left\{ 1, \frac{\lambda w_uw_v}{n} \right\} = \frac{\kappa_{uv}}{n}\text{, where }\kappa_{uv} \coloneqq \min\left\{ n, \lambda w_u w_v \right\}
\end{align} and where $\lambda \in \mathbb{R}$ is a parameter that controls the average degree.
We measure of the distance between two points using the $L_p$-norm with $1 \le p \le \infty$.
That is, we define \begin{align*}
    \|\F{x}_u - \F{x}_v\|_p \coloneqq  \begin{cases}
        \left( \sum_{i=1}^d |\F{x}_{u}(i) - \F{x}_{v}(i)|_C^p \right)^{1/p} & \text{if } p < \infty\\
        \max_i \{ |\F{x}_{u}(i) - \F{x}_{v}(i)|_C \} & \text{otherwise.}
    \end{cases}
\end{align*} where $|x - y|_C$ denotes the distance on the circle, i.e., $|x - y|_C = \min\{ |x-y|, 1 - |x-y| \}$.

Note that $L_\infty$ is a natural metric on the torus as $B_\infty(r)$, the ball of radius $r$ under this norm is a (hyper-)cube and ``fits'' entirely into $\mathbb{T}^d$ for all $0 \le r \le 1$. For this reason, the connection threshold under $L_\infty$-norm is always \begin{align*}
    t_{uv} = \frac{1}{2}\left( \frac{\lambda w_uw_v}{n} \right)^{1/d} = \left( \frac{w_uw_v}{\mu n} \right)^{1/d}
\end{align*} where $\mu = 2^d/\lambda$.

The GIRG model has a natural non-geometric counterpart where the weight distribution of the vertices is the same as in GIRGs but the edges are now sampled independently, with probability
\begin{align*}
  \Pr{u \sim v} = \min\left\{ 1, \frac{\lambda w_uw_v}{n} \right\}.
\end{align*}
This inhomogeneous random graph model is known as the Chung--Lu random graph model and has been extensively studied in literature \citep[see, e.g.,][]{acl-rgmplg-01, cl-adrgged-02, cl-ccrgg-02}. It is important to note that -- despite the fact that the connection probability of any two vertices in GIRGs and Chung--Lu graphs is the same -- these two models have important differences because edges in GIRGs do not appear independently since they further depend on the positions of the involved vertices.

Hence, for our analysis, we are now equipped with an appropriate geometric random graph model and its non-geometric counterpart. Note that, as it was shown by \citet[Theorem~1]{GirgCliques}, the two models converge as the number of dimensions in the GIRG model goes to infinity, i.e. the total variation distance of the two models goes to zero. Furthermore, we observe that the GIRG model captures many quantifiable properties of real-world networks as shown by \citet{blasius2022external}. Finally, let us note that the model is very versatile as one can consider other variants with different degree distributions or metric spaces. However, our choice of the Pareto distribution for the vertex weights and of the torus for the geometric space is the one considered most frequently in literature -- also in the results of \citet{blasius2022external}.

\paragraph{Triangles and Clustering Coefficient.}

The number of triangles and related properties of a graph are common statistics used in the analysis of networks \citep{gupta2014decompositions}, especially for detecting underlying geometry. In fact a related statistic\footnote{The statistic used, number of signed triangles, essentially measures by how much the number of triangles are in the graph exceeds the expected value in the \erdosGraph model.} was used by \citet{Bubeck_Ding_Eldan_Racz_2015} to efficiently test for the presence of geometry for a spherical random geometric graph.

When dealing with heterogeneous degree distributions, however, triangles that are attributed to large degree vertices potentially have a significant influence on the total number of triangles. In fact as shown by \citet{GirgCliques}, the number of triangles in GIRGs and in Chung--Lu graphs are asymptotically equivalent if $\beta \le 7/3$, which is not an unrealistic choice for many real-world networks. \citet{Litvak_Michielan_Stegehuis_2022}, therefore suggest weighting each triangle by the inverse degrees of the involved vertices, a statistic they call \emph{weighted triangles}. A normalized version of the number of triangles as well as cordless squares and pentagons was also used by \citet{almagro2022detecting} to determine the dimensionality of a given network.

A natural statistic, observed in many real-world networks that is however strikingly absent in non-geometric random graphs where edges are drawn independently, is the existence of a non-vanishing clustering coefficient, i.e. the probability that two randomly chosen neighbours of a vertex are adjacent. The clustering coefficient is the central focus of our analysis and we use the following common definition, also used by \citet[Definition 5.1]{Keusch_2018}.

Given a graph $G = (V, E)$, its \emph{local clustering coefficient} of a vertex $v$ is
    \begin{equation*}
        \textsc{CC}_G(v) \coloneqq \begin{cases}
            |\left\{ \{s, t\} \subseteq \Gamma(v) \mid s \sim t \right\}| / \binom{\text{deg}(v)}{2} & \text{if } \text{deg}(v) \ge 2\\
            0 & \text{otherwise.}
        \end{cases}
    \end{equation*}
The (global) \emph{clustering coefficient} of $G$ is the average of the local coefficient of each of $G$'s vertices, that is,
    \begin{equation*}
        \textsc{CC}(G) \coloneqq \frac{1}{|V|} \sum_{v\in V}\textsc{CC}_G(v).
    \end{equation*}

For GIRGs it was shown by \citet{Keusch_2018} that, when the vertices of the generated graph are drawn on a torus of constant dimensionality, the generated graph has a constant clustering coefficient. On the other hand, on Chung--Lu graphs it was shown that the clustering coefficient goes to $0$ as $n$, the number of vertices of the graph, grows \citep{Hofstad_Janssen_van_Leeuwaarden_Stegehuis_2017}. Our first result, which we discuss in the next section, extends the results on the clustering coefficient of GIRGs by giving an upper bound that explicitly depends on the dimension of the underlying space. This shows that constant dimensionality is in fact needed to obtain a constant clustering coefficient. We refine our result for the case of $L_\infty$-norm in \Cref{sec:clusteringLinftyIntro} and afterwards introduce a statistical test for learning the dimension of a network based on these results (\Cref{sec:LearningDimensionIntro}).

\subsection{Upper bounds on the clustering coefficient of GIRGs}

Our first result connects the clustering coefficient of a GIRG with the number of dimensions used to generate the positions of its vertices.
\begin{restatable}{theorem}{clusteringlp}\label{thm:clusteringLp}
    Asymptotically almost surely, if $d = o(\log(n))$, the clustering coefficient of $G$ sampled from the GIRG model under some $L_p$-norm with $p \in [1, \infty]$ is \begin{align*}
        \textsc{CC}(G) =\exp(-\Omega_d(d)) + o(1).
    \end{align*}
\end{restatable}

We remark that, for the case of $L_\infty$-norm, we later derive a sharper bound (see \Cref{thm:clusteringLinfty}). \Cref{thm:clusteringLp} implies that if $d = \omega(1)$ and $d = o(\log(n))$ the clustering coefficient vanishes. As most real-world networks have a non-vanishing clustering coefficient, our theorem suggests that their dimensionality must be at most constant in the number of vertices. This can be seen as a theoretical explanation for the empirical observations of the low dimensionality of real-world networks by~\cite{almagro2022detecting}.

Besides the results of \citet{Keusch_2018} for a constant number of dimensions, the clustering coefficient of random geometric graphs (i.e., our model in the case of homogeneous weights) under the $L_2$-norm as a function of $d$ was previously analysed by \citet{Dall_Christensen_2002}. However, our \Cref{thm:clusteringLp} also applies to inhomogeneous degree distributions and arbitrary $L_p$-norms, which complicates the analysis. The main difficulty in proving \Cref{thm:clusteringLp} is that the probability that two random neighbours of a given vertex are connected is significantly influenced by their weights. To circumvent this issue we show that high-weight vertices only have a small influence on the global clustering coefficient of a power-law graph $G$ in \Cref{apx:bounding}. Via an application of the method of typical bounded differences (\Cref{thm:bounded-diff} see also the article by~\cite{warnke2016method}) -- a generalisation of McDiarmid's inequality \cite{mcdiarmid1989method} and a powerful tool to showing concentration in high dimensional spaces -- we then show that the clustering coefficient of a GIRG concentrates around the expected clustering coefficient of a subgraph induced by vertices of small weight.

The bound on the clustering coefficient of the low-weight vertex subgraph follows from a bound on the probability that two random vectors $\F{y}_u, \F{y}_v$ uniformly distributed within the ball of radius $1$ have a distance larger than a certain threshold. Intuitively, the fact that this probability decays exponentially in $d$ is a consequence of the law of large numbers: as $d$ grows, with large probability, about half of the components of $\F{y}_u$ and $\F{y}_v$ have opposite sign, which already leads to a distance between $u$ and $v$ that is arbitrarily close to 1 with probability converging to 1 as $d$ grows. Taking into account that the other components of $\F{y}_u$ and $\F{y}_v$ also contribute at least a constant increase in distance between $u$ and $v$ with large probability, we get that there is an exponentially increasing probability that the distance between $u$ and $v$ is strictly greater than one, which suffices to show an exponential upper bound on the clustering coefficient in $G$. To prove this exponential decay in terms of $d$, we use a coupling argument based on the observation that the ``direction'' $\F{x}/\|\F{x}\|_p$ and the norm $\|\F{x}\|_p$ of a random vector distributed in the unit ball under $L_p$-norm are independent. To analyze the normalized vector $\F{x}/\|\F{x}\|_p$, we define the following distribution and show that if $\F{z}$ is a vector sampled from this distribution, then $\F{z}/\|\F{z}\|_p$ is distributed just as $\F{x}/\|\F{x}\|_p$. This has the advantage that the components of $\F{z}$ are now independent, allowing us to apply sharp tail bounds from which our statements follow.

\paragraph{The $\rchi^p$-Distribution.}
Let $p\in\mathbb{R}, p \ge 1$. We call a random vector $\F{x} \in \mathbb{R}^d$, $\rchi_p(d)$ distributed if each of its components $\F{x}(i)$ is independently distributed according to the density function \begin{equation*}
        \rho(x_i) \coloneqq \gamma e^{-\frac{1}{2}\abs{\F{x}(i)}^p}
    \end{equation*} with the normalising constant \begin{align*}
    \gamma = \frac{p}{2^{1/p + 1}\Gamma(1/p)},
    \end{align*} where $\Gamma(s) = \int_0^\infty x^{s-1}e^{-x} \d x$ is the gamma function.
    If $\F{x} \sim \rchi_p(d)$, then we denote the distribution of the random variable $\left(\norm{\F{x}}_p\right)^p = \sum_{i=1}^d |\F{x}(i)|^p$ by $\rchi^p(d)$.

This distribution is a generalisation of the $\rchi^2$ distribution and a simplification of the one proposed by \citet{Livadiotis_2014}. In our analysis, we determine its moment generating function. This not only gives us its expectation, which is $2d/p$, but also allows us to obtain the following concentration bound, which we use in the arguments used for the proof of \Cref{thm:clusteringLp}.

\begin{restatable}{corollary}{tailchisquared}\label[corollary]{cor:tailchisquared}
    Let $X_i, \ldots X_d$ be i.i.d. random variables from $\rchi_p(1)$ and define $Z = \sum_{i=1}^d|X_i|^p \sim \rchi^p(d)$. Then, for every $\varepsilon > 0$,
    \begin{align*}
        &\Pr{|Z - \Expected{Z}| \ge \varepsilon \cdot \Expected{Z}} \le 2 \exp \left( -\frac{2\delta}{p} d \right).
    \end{align*}
    Where $\delta > 0$ is defined by $\varepsilon =  2(\sqrt{2\delta}  + \delta) $.
\end{restatable}

We believe our analysis of the $\rchi_p(d)$ and $\rchi^p(d)$ distributions to be of independent interest, as many random spaces can be related to vectors drawn uniformly at random within the $d$-dimensional unit ball of some $L_p$-norm.

\paragraph{Improved bounds for the $L_\infty$-norm.}\label{sec:clusteringLinftyIntro}

When using $L_\infty$-norm as a distance measure for GIRGs we obtain more precise results and are able to further determine the base of the exponential function governing the decay of $\textsc{CC}(G)$. Recall that the $L_\infty$ norm is not only a natural distance measure on the torus from a mathematical point of view, but also one that yields graphs that closely resemble real-world networks \citep{blasius2022external}.

\begin{restatable}{theorem}{clusteringlinfty}\label{thm:clusteringLinfty}
    Assume that $d = o(\log(n))$ and $\beta \neq 3$. Then asymptotically almost surely, the clustering coefficient of $G$ sampled from the GIRG model with $L_\infty$-norm fulfils
    \begin{align*}
        \textsc{CC}(G) = \mathcal{O}_d\left( \left(\frac{3}{4}\right)^{\min\{1, \beta - 2\} d} \right) + o(1)
    \end{align*} and \begin{align*}
        \textsc{CC}(G) = \Omega_d\left( \max\left\{ \left( \frac{1}{4} \right)^{(\beta-2)d}, \left( \frac{3}{4} \right)^d \right\} \right).
    \end{align*} In particular, if $\beta > 3$, we have \begin{align*}
        \textsc{CC}(G) = \Theta_d\left( \left( \frac{3}{4} \right)^d \right) + o(1).
    \end{align*}
\end{restatable}

This theorem shows that $\textsc{CC}(G)$ essentially decays as $(3/4)^d$ asymptotically in $d$ if $\beta > 3$. Otherwise, if $\beta \in (2,3)$, we obtain slightly weaker bounds but we are in particular able to show that $\textsc{CC}(G)$ decays asymptotically slower than $(3/4)^d$ if $\beta$ is sufficiently close to $2$, which follows from the lower bound $\textsc{CC}(G) = \Omega_d\left( (1/4)^{(\beta-2)d} \right)$. The reason for this is that the expected weight of a random neighbor of a given vertex is infinite if $\beta \in (2,3)$ which leads to an increased overall clustering coefficient. 

The proof of \Cref{thm:clusteringLinfty} is based on an application of the following theorem by \citet[Theorem~3]{GirgCliques} that bounds the probability that a set of $k$ vertices forms a clique conditioned on the event that said vertices form a star centered at the vertex of minimal weight under the assumption that the ratio between the minimal and maximal weight is bounded. We slightly reformulate the original statement for the sake of exposition.

\begin{restatable}{theorem}{sharpcliquebounds}\label{thm:sharpcliquebounds}
    Let $G$ be a GIRG generated under $L_\infty$-norm. Let $U = \{ v, s, t \}$ be a set of $3$ vertices with weights $w_v, w_s, w_t$ such that $w_v \le w_s \le w_t$ and $w_v \le c w_t$ for some constant $c>0$. If $\left( w_t^2/(\mu n) \right)^{1/d} \le 1/4$, we have \begin{align*}
        \left(\frac{3}{4}\right)^{d} \le \Pr{U \text{ is a triangle} \mid v \sim s,t}& \le c \left(\frac{3}{4}\right)^d.
    \end{align*}
\end{restatable}        

We remark that the condition $\left( w_t^2/(\mu n) \right)^{1/d} \le 1/4$ is needed to ensure that the connection threshold for $s,t$ is sufficiently small such that we can ignore the topology of the underlying torus and measure distances as in $\mathbb{R}^d$.

\subsection{Testing for the dimensionality}\label{sec:LearningDimensionIntro}

A natural further question that arises is whether one can recover the underlying dimension of a given GIRG by means of statistical testing. The previous secons suggest that the clustering coefficient is an indicator of this property, however, we have also seen that this metric is further influenced by other model parameters (i.e. $w_0, \lambda$ and especially $\beta$) making it rather unsiutable for designing a rigorous test. A similar problem arises when using the total number of triangles, which is is dominated by those forming among large degree vertices independently of $d$ if $\beta$ is close to $2$ as observed by \cite{GirgCliques,Litvak_Michielan_Stegehuis_2022}. In \cite{Litvak_Michielan_Stegehuis_2022}, the authors therefore suggest to count the number of weighted triangles instead, where each triangle contributes a weight that is inversely proportional to the product of the degrees of its vertices. Weighted triangles thus counteract the effect of large degree vertices as the influence of triangles forming among such vertices is diminished. However, this approach only allows to decide whether the network has an underlying metric structure, but not its dimensionality.

We take a similar (yet more direct) approach for excluding the effect of large degree vertices and introduce a test that is further able to infer the dimension of the underlying metric space. Namely, we can show that the clustering coefficient among vertices of approximately the same weight that have at least two neighbors is highly concentrated and a direct indicator of the underlying dimension, without being influenced by other model parameters. More precisely, using \Cref{thm:sharpcliquebounds} together with the method of typical bounded differences (\Cref{thm:bounded-diff}), we can show that the average local clustering coefficient in the induced subgraph of all vertives with weight in some interval $[w_c, cw_c]$ (where $w_c \ge w_0, 0 < c < 2/\sqrt{3}$) concentrates tightly around a value that is only dependent on $d$ and not on $\beta$ or $w_0$. This is formalised in the following theorem.

\begin{restatable}{theorem}{dimensionalitytesting}\label{thm:dimensionalitytesting}
    Let $G = G(n, d, \beta, w_0)$ be a GIRG generated under $L_\infty$-norm. Let further $1 < c < 2/\sqrt{3} \approx 1.1547, w_c \ge w_0$ be constants, and let $\Gtest$ be the subgraph of $G$ consisting of all vertices with weight in $[w_c, c  w_c]$. Assume that $d$ is an integer with $d = o(\log(n))$. Define the set $S$ as the set of vertices in $\Gtest$ that have at least two neighbors in $\Gtest$ and the random variable $\teststatistic{\Gtest}$ as \begin{align*}
        \teststatistic{\Gtest} \coloneqq \frac{1}{|S|} \sum_{v \in S} \frac{|\left\{ \{s,t\} \in \Gamma(v) \mid s \sim t \right\}|}{\binom{\deg(v)}{2}} = \frac{1}{|S|} \sum_{v \in S} \textsc{CC}_{\Gtest}(v).
    \end{align*} Then, \begin{align}\label{eq:testcondition}
        \teststatistic{\Gtest} \in \left( \frac{1}{c}\left( \frac{3}{4} \right)^d, c\left( \frac{3}{4} \right)^d \right) \pm n^{-1/5}
    \end{align} with probability at least $1 - 1/n$.
\end{restatable}

\Cref{thm:dimensionalitytesting} can be viewed as a linear-time algorithm for the following statistical testing problem (assuming that $w_c$ is constant). We are given a graph $G$ on $n$ vertices, its weight sequence, and an integer $d = o(\log(n) )$. Under the null hypothesis, $G$ is a GIRG generated with dimension $d$, whereas under the alternative hypothesis, $G$ was generated in dimension $d_1 \neq d$ or it is a Chung--Lu graph. Here, we allow $d_1$ to be any integer (potentially larger than $\log(n)$). Consider the following testing procedure for this problem. Fix a constant $1 < c <2/\sqrt{3}$ and a weight $w_c \ge w_0$. Now, consider the induced subgraph $\Gtest$ of $G$ consisting of all nodes with weight in $[w_c, c w_c]$. For every node $v \in \Gtest$ that has at least two neighbours in $\Gtest$, compute its local clustering coefficient $\textsc{CC}_{\Gtest}(v)$ and denote by $\teststatistic{\Gtest}$ the mean over all these values. We accept the null hypothesis if and only if condition (\ref{eq:testcondition}) is met. Due to \Cref{thm:dimensionalitytesting} the probability that this test is incorrect under both the null and alternative hypothesis goes to zero as $n \rightarrow \infty$. Furthermore, the running time of this test is linear, as we have to compute the local clustering coefficient of vertices of constant weight and as the degree of a vertex with weight $c w_c$ is constant in expectation. Iterating this statistical test over the range of $d$ we can recover the dimensionality of the input graph with high probability. Let us note that our result is not restricted to a constant number of dimensions but applies to the whole regime $d = o(\log(n))$, which -- as \Cref{thm:clusteringLinfty} implies -- is the only relevant one for this problem.

\subsection{Application to real-world networks.}

In addition to our theoretical results, we tested our algorithm in practice, both in real world networks and in GIRGs. For estimating the vertex weights, we used the maximum likelihood estimator derived in \cite[Appendix B.2]{Boguna_Papadopoulos_Krioukov_2010}.
The outcome of our experiments is summarised in \Cref{fig:experiments}. \Cref{tab:graphs} further contains a list of the real-world networks we used for the first row of plots in \Cref{fig:experiments}. In \Cref{fig:experiments}, the size of the circles is proportional to the number of vertices in the induced subgraph of vertices with weight in the interval $[w_c, cw_c]$. We use $c = 1.155$ which is roughly the maximum permitted value predicted by \Cref{thm:dimensionalitytesting}. The dashed lines represent roughly the expected value of our test statistic for a GIRG genenerate in dimension $d$, i.e. $(3/4)^d$. The GIRGs were generated with the sampling algorithm of \citet{Blaesius_Friedrich_Katzmann_Meyer_Penschuck_Weyand_2019} using $\beta = 2.5, \alpha=10$ and an average degree of $10$. The histogram on the lower right of \Cref{fig:experiments} shows the frequency of each inferred dimension using the dataset of \citet{blasius2022external} consisting of 2976 real-world networks. The dimension here is inferred by taking the weighted median of the inferred dimension from our test statistic over different values of $w_c$ ranging from $2$ to $300$. The weighting is by the number of vertices in the respective subgraph induced by vertices with weight in $[w_c, cw_c]$. 

The inferred dimensions we obtain are indeed similar to the results of \citet[Fig.~5]{almagro2022detecting} with typical inferred dimensions being in the range of 1 to 10 and with social networks being generally assigned to higher dimensions than collaboration or citation networks. Some of the networks we use (all the ones we were able to find) are also contained in the dataset of \cite{almagro2022detecting} and here, the inferred dimensions of both approaches are very similar. For examples, consider \texttt{email-Enron, ca-AstroPh, ca-CondMat, ca-GrQc, cit-HepTh}.

We remark that the big advantage of our algorithm is that it has stronger theoretically foundations and is much more efficient. In fact, we are able to handle datasets of orders of magnitude larger than \citet{almagro2022detecting}. In fact, our experiments conducted on the set of 2976 real-world networks by \citet{blasius2022external} which are summarized in the histogram in \Cref{fig:experiments} show that a vast majority of them is assigned a dimension of at most 10. This can be seen as a further indication of the ultra-low-dimensional nature of most realistic networks, now tested on a much larger dataset as before. 

Besides that, we note that the algorithm works exceptionally well for synthetic networks and this holds even if we use a soft version of GIRGs which includes an additional temperature parameter $\alpha \ge 1$ that (if close to $1$) diminishes the influence of the underlying geometry (we refer to \cite{Blaesius_Friedrich_Katzmann_Meyer_Penschuck_Weyand_2019} for more information). Experiments indicate that the test continues to achieve a performance similar to that predicted by theory for all $\alpha \ge 2$. However, even for smaller values of $\alpha$, the inferred dimension of our test can still be seen as an upper bound on the ground truth since it is known that smaller values of $\alpha$ only lead to a decrease in clustering. 

We further observe that some of the considered real-world networks show an overall similar behaviour as that of the GIRGs (e.g. \texttt{soc-academia, fb-pages-artist, ca-AstroPh, ca-CondMat}). It is however not a surprise that real-world data can be noisy and, therefore, exhibit a  behaviour that differs from GIRGs. Similar difficulties were also encountered by~\citet{almagro2022detecting} (see their supplementary material) and a similar noisy behaviour can also be observed in small generated GIRGs, where number of vertices is not hight enough for the concentration results to be strong. Nevertheless, some of the considered networks (especially the biological networks) show a rather different behavior as predicted by the GIRG model which indicates that GIRGs do not capture all properties of realistic networks and thus motivates further research.

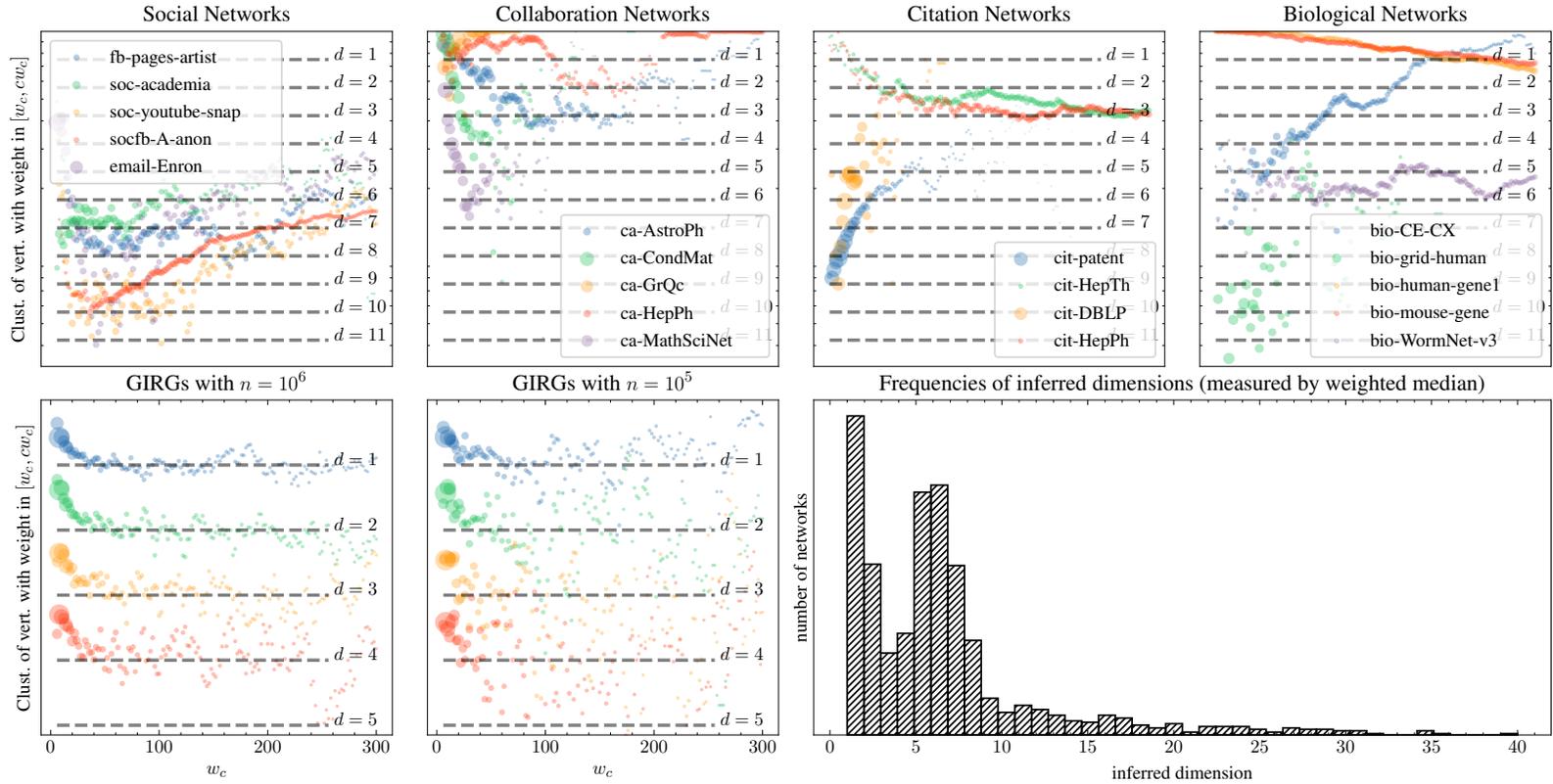
\begin{landscape}
\begin{figure}
    \centering
    \resizebox{1.41\textwidth}{!}{ \input{plot_with_new_constant_and_fixed_Enron.pgf} }
    \caption{The clustering coefficient of the low-weight vertices for different choices of $w_c$ in real-world and synthetic networks. The size of the circles is proportional to the number of vertices in the induced subgraph of vertices with weight in the interval $[w_c, cw_c]$ whereby we use $c = 1.155$. The dashed lines represent roughly the expected value of our test statistic for a GIRG. The histogram shows the frequency of each inferred dimension using the dataset of \citet{blasius2022external} consisting of 2976 real-world networks. 
    \label{fig:experiments}}
\end{figure}
\end{landscape}

\subsection{Future work}

As we previously discussed, a large body of work has been devoted to understanding in which cases (i.e. for which asymptotic behaviors of $d$), geometry is detectable in spherical random geometric graphs (SRGGs) for homogeneous weights. While the parameter regime where these graphs lose their geometry in the dense case, i.e. the case where the marginal connection probability of two vertices is constant and does not depend on $n$, is well understood \citep{Devroye_Gyoergy_Lugosi_Udina_2011,Bubeck_Ding_Eldan_Racz_2015,liu2021probabilistic}, it remains unclear what happens in the sparse case (where the marginal connection probability is proportional to $1/n$) and progress has been made only recently \citep{brennan2020phase,Liu2022TestingThresholds}.

On the other hand, there has not been much research devoted to studying the influence of the dimension on random geometric graphs in the case of inhomogeneous weights. We gave first results in this regard by studying how the clustering coefficient depends on $d$ and showed that the dimension can be detected by means of statistical testing assuming that $d = o(\log(n))$. It remains to study under which conditions the geometry remains detectable if $d = \Omega(\log(n))$ and under which circumstances the model converges to its non-geometric counterpart with respect to the total variation distance of the distributions over the produced graphs as previously studied for spherical random graphs. Furthermore, it remains to study what differences arise when using the torus instead of the sphere as the underlying metric space. We expect that our model loses its geometry earlier than spherical random graphs, as the number of triangles in our model for the sparse case with constant weights is, in expectation, the same as in an Erdős-Rényi graph already if $d = \omega(\log^{3/2}(n))$ \citep{GirgCliques}. On the sphere this only happens if $d = \omega(\log^3(n))$ \citep{Bubeck_Ding_Eldan_Racz_2015}.

For detecting the geometry in SRGGs \citet{Bubeck_Ding_Eldan_Racz_2015} have introduced the \emph{signed triangle} statistic which proves more powerful than ordinary triangles in the dense case. It remains to study if signed triangles, or a combination of signed triangles and weighted triangles considered by \citet{Litvak_Michielan_Stegehuis_2022}, gives rise to a more powerful test for the case of inhomogeneous weights.

A further interesting extension is to consider an anisotropic version of our model, along the lines of the work of \citet{eldan2020information,brennan2022threshold}. In the GIRG model, this can be naturally encoded in the distance measure used to determine the edge threshold.

Noisy settings have also been considered in the context of testing for geometry in random graphs \citep{liu2021phase}. Our model comes with a natural way of modelling noise in the form of and inverse temperature parameter $\alpha>1$ \citep{Keusch_2018}. Here, two vertices are connected with a probability that depends on both their distance and the temperature. More precisely, the connection probability of two vertices $u,v$ fulfills \begin{align*}
    p_{uv} = \Theta\left( \frac{1}{\|\F{x}_u - \F{x}_v\|_\infty^{\alpha d}} \left( \frac{w_uw_v}{n} \right)^\alpha \right).
\end{align*} Intuitively, lower values of $\alpha$ diminish the influence of the underlying geometry. We expect results similar to ours to hold in the noisy setting if $\alpha$ is a constant. It remains to study how different (constant or superconstant) values of $\alpha$ influence the detectability of the underlying geometry and dimension.

\section{Preliminaries}

We let $G = (V, E)$ be a (random) graph on $n$ vertices. For any value $\delta$, let $V_{\leq \delta}$ denote the set of vertices of degree at most $\delta$ and let $G_{\leq \delta}$ denote the subgraph of $G$ induced by $V_{\leq \delta}$.
We use standard Landau notation to describe the asymptotic behavior of
functions for sufficiently large $n$. That is, for functions $f,g$, we
write $f(n) = \mathcal{O}(g(n))$ if there is a constant $c > 0$ such
that for all sufficiently large $n$, $f(n) \le c g(n)$. Similarly, we
write $f(n) = \Omega(g(n))$ if $f(n) \ge c g(n)$ for sufficiently
large $n$. If both statements are true, we write
$f(n) = \Theta(g(n))$. Regarding our study of the clustering
coefficient, some results make a statement about the asymptotic
behavior of a function with respect to a sufficiently large $d$. These
are marked by
$\mathcal{O}_d(\cdot), \Omega_d(\cdot), \Theta_d(\cdot)$,
respectively.

\subsection{Probabilistic tools}

We say that an event $\mathbf{E}$ holds asymptotically almost surely if $\lim_{n\rightarrow\infty}\Pr{\mathbf{E}}=1$ and with high probability if $\Pr{\mathbf{E}}\geq 1 - \mathcal{O}(1/n)$.
The following theorem shows that the sum of independent Bernoulli random variables converges to a Poisson distributed random variable if the individual success probabilities are small.

\begin{theorem}[Proposition 1 in \cite{Cam_1960}]\label[lemma]{thm:lecam}
    For $1 \le i \le n$, let $X_i$ be independent Bernoulli distributed random variables such that $\Pr{X_i = 1} = p_i$. Let $\lambda_n = \sum_{i=1}^n p_i$, and $S = \sum_{i=1}^n X_i$. Then, \begin{align*}
        \sum_{k=0}^\infty \left| \Pr{S = k} - \frac{\lambda_n^ke^{-\lambda_n}}{k!} \right| \le 2 \sum_{i=1}^k p_i^2.
    \end{align*} 
\end{theorem}

We will also use the following concentration bounds.
\begin{theorem}[Theorem 2.2 in \cite{Keusch_2018}, Chernoff-Hoeffding Bound]\label{thm:chernoff-hoeffding}
    For $1 \le i \le k$, let $X_i$ be independent random variables taking values in $[0,1]$, and let $X \coloneqq \sum_{i=1}^k X_i$. Then, for all $0 < \varepsilon < 1$,
    \begin{enumerate}
        \item[(i)] $\Pr{X > (1 + \varepsilon)\Expected{X}} \le \exp \left( -\frac{\varepsilon^2}{3}\Expected{X} \right)$.
        \item[(ii)] $\Pr{X < (1 - \varepsilon)\Expected{X}} \le \exp \left( -\frac{\varepsilon^2}{2}\Expected{X} \right)$.
        \item[(iii)] $\Pr{X \ge t} \le 2^{-t}$ for all $t \ge 2e\Expected{X}$. 
    \end{enumerate}
\end{theorem}
While this theorem is extremely useful when dealing with sums of independent random variables, we shall further need the \emph{method of typical bounded differences} to obtain bounds when the Chernoff-Hoeffding bound is not applicable. 
\begin{theorem}[Theorem 2.5 in \cite{Keusch_2018}]\label{thm:bounded-diff} 
    Let $X_1, \ldots, X_m$ be independent random variables over $\Omega_1, \ldots, \Omega_m$. Let $X = (X_1, \ldots, X_m) \in \Omega = \prod_{i=1}^m \Omega_i$ and let $f: \Omega \rightarrow \mathbb{R}$ be a measurable function such that there is some $M > 0$ such that for all $\omega \in \Omega$, we have $0 \le f(\omega) \le M$. Let $\mathcal{B} \subseteq \Omega$ such that for some $c > 0$ and for all $\omega, \omega' \in \overline{\mathcal{B}}$ that differ in at most two components $X_i, X_j$, we have \begin{equation*}
    \abs{f(\omega) - f(\omega')} \le c.
\end{equation*}
Then, for all $t \ge 2M\Pr{\mathcal{B}}$, we have \begin{equation*}
    \Pr{\abs{f(X) - \Expected{f(X)}} \ge t} \le 2 \exp \left( -\frac{t^2}{32mc^2} \right) + \left(\frac{2Mm}{c} + 1\right)\Pr{\mathcal{B}}.
\end{equation*}
\end{theorem}

\subsection{Basic properties of the GIRG model}

We will need the following statements about the distribution of the degrees and weights in the GIRG model.
\begin{lemma}[Lemma 3.3 and Lemma 3.4 in \cite{Keusch_2018}, slightly reformulated]\label[lemma]{lem:degrees}
    The following properties hold for $G(n, d, \beta, w_0) = (V, E)$.
    \begin{enumerate}
        \item[(i)] For all $v \in V$, we have $\Expected{\deg(v)} = \Theta(w_v)$.
        \item[(ii)] With probability $1 - n^{-\omega(1)}$, we have for all $v \in V$ that $\deg(v) = \mathcal{O}(w_v + \log^2(n))$. 
    \end{enumerate}
\end{lemma}
\noindent In fact, we need a slightly stronger version of statement (ii) above. \begin{lemma}\label[lemma]{lem:strongdegreebound}
    Let $V_{\le \log(n)}$ be the set of all vertices with weight at most $\log(n)$.
    With probability at least $1 - n^{-\Omega(\log^2(n))}$, we have for all $v \in V_{\le \log(n)}$ that $\deg(v) \le \log^3(n)$.
\end{lemma} \begin{proof}
    We get from \Cref{lem:degrees} that $\Expected{\deg(v)} = w_v$. Hence, for sufficiently large $n$, we get that $\log^3(n) \ge 2e\Expected{\deg(v)}$ for all $v \in V_{\le \log(n)}$. Since the degree of a fixed vertex $v$ conditioned on its position is a sum of independent Bernoulli distributed random variables, we may apply statement (iii) from \Cref{thm:chernoff-hoeffding} to obtain $\Pr{\deg(v) \ge \log^3(n) } \le n^{-\Omega(\log^2(n))}$. From a union bound, we get that the probability that at least one vertex from $V_{\le \log(n)}$ has a degree of $\log^3(n)$ or more is at most $n \cdot n^{-\Omega(\log^2(n))} = n^{-\Omega(\log^2(n))}$, which concludes the proof.
\end{proof}

\section{Upper bound on the clustering coefficient of GIRGs}

We proceed by pointing out our general bounding technique and then handle the case of $L_\infty$-norm and $L_p$-norms with $p \in [1,\infty)$ separately.

\subsection{Our bounding technique}\label{apx:bounding}

We go on with developing a technique for upper bounding $\textsc{CC}(G)$. The main difficulty here is that the probability that two random neighbors of a given vertex are connected grows significantly with their weight. We circumvent this issue by showing that high-weight vertices only have a small influence on the global clustering coefficient of a power-law graph $G$, which essentially concentrates around its expectation in an induced subgraph of small weight. We formalize this in the following lemma that is proved in a similar way as \cite[Theorem~4.4]{Keusch_2018}. 

\begin{lemma}\label[lemma]{lem:clusteringbound}
    Asymptotically almost surely, we have \begin{align*}
        \textsc{CC}(G) = \Expected{\textsc{CC}( G_{\le n^{1/8}})} + o(1).
    \end{align*}
\end{lemma}

To prove this statement, we require the following auxiliary lemmas.
\begin{lemma}[Lemma~3.5 in \cite{Keusch_2018}]\label[lemma]{lem:weightsampling}
    If the weight $w$ of each vertex is sampled from the Pareto-distribution with parameters $w_0,1-\beta$, then for all $\eta > 0$, there is a constant $c > 0$ such that with probability $1 - n^{-\omega\left(\log \log(n)/\log(n)\right)} = 1 - o(1)$,  and all $w \ge w_0$, we have \begin{align*}
        |V_{\ge w}| \le c n w^{1 + \eta - \beta}.
    \end{align*} 
\end{lemma}

\begin{lemma}[Lemma~3.4 in \cite{Keusch_2018}]\label[lemma]{lem:degreebound}
   With probability $1-n^{-\omega(1)}$, for all $v \in V$, we have $\deg(v) = \mathcal{O}(w_v + \log(n)^2)$.
\end{lemma}

\begin{lemma}[Lemma 3.1 in \cite{Keusch_2018}]\label[lemma]{lem:weightsum}
    If for all $\eta > 0$, there is a constant $c> 0$ such that for all $w \ge w_0$, we have $|V_{\ge w}| \le c n w^{1 + \eta - \beta}$, then \begin{align*}
        \sum_{v \in V_{\ge w}} w_v = \mathcal{O}(nw^{2+\eta - \beta}).
    \end{align*}
\end{lemma}

\begin{proof}[Proof of \Cref{lem:clusteringbound}]
    We start by showing that \begin{align*}
        \textsc{CC}(G_{\le n^{1/8}}) = \Expected{\textsc{CC}(G_{\le n^{1/8}})} + o(1)
    \end{align*} asymptotically almost surely and then how this statement transfers to the whole graph $G$. 

    To show concentration, we use \Cref{thm:bounded-diff} and note that the positions and weights of all vertices define a product probability space as in \Cref{thm:bounded-diff}. We denote this space by $\Omega$, whereby every $\omega \in \Omega$ defines a graph $G(\omega)$ on the vertex set $V_{\le n^{1/8}}$. Note that the number of independent random variables is $m = 2n$. Thus, we may define a function $f: \Omega \rightarrow \mathbb{R}$ that maps every $\omega \in \Omega$ to $\textsc{CC}(G(\omega))$.
    We consider the "bad" event \begin{equation*}
        \mathcal{B} = \{ \omega \in \Omega \mid \text{the maximum degree in $G(\omega)$ is at least $n^{1/4}$} \}.
    \end{equation*}
    By \Cref{lem:degreebound}, we get that $ \Pr{\mathcal{B}} = n^{-\omega(1)}$. Now, let $\omega, \omega' \in \overline{\mathcal{B}}$ such that they differ in at most two coordinates. We observe that changing the weight or coordinates of one vertex $v$ only influences the clustering coefficient of $v$ itself or vertices that are neighbors of $v$ before or after the change. Since $v$ has at most $n^{1/4}$ neighbors in both $G(\omega)$ and $G(\omega')$, the change affects at most $2n^{1/4}$ vertices. Two such changes can hence only increase or decrease the clustering coefficient of $G(\omega)$ by at most $4n^{1/4}/n$, and so we have $\abs{f(\omega) - f(\omega')} \le 4n^{-3/4}$. We note that the choice $t = n^{-1/8}$ fulfills the condition $t \ge 2M\Pr{\mathcal{B}}$ since $M = 1$ and $\Pr{\mathcal{B}} = n^{-\omega(1)}$. Thus, we may apply \Cref{thm:bounded-diff} to obtain \begin{align*}
        & \Pr{|\textsc{CC}(G_{\le n^{1/8}}) - \Expected{\textsc{CC}(G_{\le n^{1/8}})} | \ge n^{-1/8}} \\ &\hspace{1cm} \le 2 \exp\left(-\frac{n^{-1/4}}{32 \cdot 2 n \cdot 16 n^{-3/2}}\right) + \left(\frac{4n}{n^{-3/4}} + 1\right)n^{-\omega(1)} = n^{-\omega(1)}.
    \end{align*} This shows that with high probability, $|\textsc{CC}(G_{\le n^{1/8}}) - \Expected{\textsc{CC}(G_{\le n^{1/8}})}| = o(1)$.

    In order to transfer this finding to the entire graph $G$, we note that each additional vertex we add to $G_{\le n^{1/8}}$ has (local) clustering of at most one and each edge, we add to a vertex $v \in V_{\le n^{1/8}}$ can only increase its clustering by at most one as well. Hence, \begin{align*}
        \textsc{CC}(G) &\le \frac{1}{n} \left( |V_{\le n^{1/8}}| \textsc{CC}(G_{\le n^{1/8}}) + |V_{> n^{1/8}}| + \sum_{v \in V_{> n^{1/8}}} \deg(v) \right)\\
        &\le \textsc{CC}(G_{\le n^{1/8}}) + \frac{|V_{> n^{1/8}}|}{n} + \frac{1}{n}\sum_{v \in V_{> n^{1/8}}} \deg(v).
    \end{align*}
    To bound this term, we note that the probability that a random vertex $v$ has weight greater than $n^{1/8}$ is proportional to $n^{(1-\beta)/8} = o(n^{-1/8})$. Hence, the expected size of $V_{> n^{1/8}}$ is $o(n^{7/8})$ and by a Chernoff bound, we get that $|V_{> n^{1/8}}| \le 2 \Expected{V_{> n^{1/8}}}$ with high probability, implying $|V_{> n^{1/8}}|/n = o(1)$ with high probability By \Cref{lem:degreebound}, we further get that $\deg(v) = \mathcal{O}(w_v)$ for all $v \in V_{> n^{1/8}}$ and hence, by \Cref{lem:weightsum} and \Cref{lem:weightsampling}, we get \begin{align*}
        \sum_{v \in V_{> n^{1/8}}} \deg(v) = \mathcal{O} \left( \sum_{v \in V_{> n^{1/8}}} w_v \right) = \mathcal{O}(n^{1 + (2+\eta - \beta)/8}) = o(n)
    \end{align*} asymptotically almost surely and for some sufficiently small $\eta > 0$ from which our statement follows.
\end{proof}

We further require the following lemma, which formalizes that the clustering coefficient of a vertex $v$ can equally be seen as the probability that two randomly chosen neighbors of $v$ are adjacent. 

\begin{lemma}\label[lemma]{lem:triangleclustering}
    Let $v, s, t$ be three vertices from $G$, chosen uniformly at random. Denote by $\Delta$ the event that $v, s, t$ form a triangle. We have \begin{align*}
        \Expected{\textsc{CC}(G)} = \Pr{\Delta \mid v \sim s, t}\Pr{\deg(v) \ge 2}.
    \end{align*} 
    Furthermore, let $\hat{v}, \hat{s}, \hat{t}$ be the vertices $v, s, t$ ordered increasingly by their weights. Then, \begin{align*}
        \Expected{\textsc{CC}(G)} \le \Pr{\Delta \mid \hat{v} \sim \hat{s}, \hat{t}}.
    \end{align*}
\end{lemma} \begin{proof}
    We start by showing the first statement. Assume that $V = \{u_1, \ldots, u_n\}$ and observe that, by linearity of expectation, \begin{align*}
        \Expected{\textsc{CC}(G)} = \frac{1}{n}\sum_{u \in V}\Expected{\textsc{CC}_G(u)} =  \Expected{\textsc{CC}_G(u_1)}
    \end{align*} as every vertex has the same expected local clustering assuming that its weight is an independent sample from the Pareto distribution. It thus suffices to show that $\Expected{\textsc{CC}_G(v)} \le \Pr{\Delta \mid v \sim s, t }$. For this, recall that $\Gamma(v) = \{u_1, \ldots, u_k\}$ is the (random) set of neighbors of $v$ numbered from $1$ to $k$ in some random order. Observe that $\deg(v) = |\Gamma(v)|$ and recall that the random variable $\textsc{CC}_G(v)$ is defined as \begin{align*}
        \textsc{CC}_G(v) = \frac{1}{\binom{|\Gamma(v)|}{2}} \sum_{i < j} \mathds{1}(u_i \sim u_j),
    \end{align*} where $\mathds{1}(s \sim t)$ is an indicator random variable that is $1$ if and only if $s$ and $t$ are connected. By linearity of expectation, we get that, for any $k \ge 2$, \begin{align*}
        \Expected{\textsc{CC}_G(v) \mid |\Gamma(v)| = k} = \frac{1}{\binom{k}{2}} \sum_{i < j} \Pr{u_i \sim u_j \mid \deg(v) = k}.
    \end{align*}
    We proceed by showing that for any $1 \le i < j \le k$, we have $\Expected{\mathds{1}(u_i \sim u_j) \mid \deg(v) = k } = \Pr{s \sim t \mid s,t \in \Gamma(v) }$. To this end, let $\Omega$ be the global sample space consisting of all possible $n$-vertex graphs and two of its vertices $s, t$ chosen u.a.r. Let further $\mathcal{B} \subset \Omega$ be the set of all outcomes where $\deg(v) = k$ and where $s = u_i$ and $t = u_j$. We have, \begin{align*}
        \Expected{\mathds{1}(u_i \sim u_j) \mid \deg(v) = k } &= \Pr{u_i \sim u_j \mid \deg(v) = k}\\
        &= \Prsub{\mathcal{B}}{s \sim t}\\
        &= \Prsub{\Omega}{s \sim t \mid \mathcal{B}}\\ 
        &= \Pr{ s \sim t \mid (s = u_i) \cap (t = u_j) \cap (\deg(v) = k) } \\
        &= \frac{ \Pr{ (s \sim t) \cap (s = u_i) \cap (t = u_j) \cap (\deg(v) = k) \mid s, t \in \Gamma(v) } }{ \Pr{ (s = u_i) \cap (t = u_j) \cap (\deg(v) = k) \mid s,t \in \Gamma(v) }}\\
        & = \Pr{s \sim t \mid s, t \in \Gamma(v) },
    \end{align*} where the second to last equality holds because the events $s \sim t$ and $s = u_i \cap t = u_j \cap \deg(v) = k$ are independent if we condition on $s, t \in \Gamma(v)$.
    This implies \begin{align*}
        \Expected{\textsc{CC}_G(v) \mid \deg(v) = k} &= \frac{1}{\binom{k}{2}} \sum_{i < j} \Pr{u_i \sim u_j \mid |\Gamma(v)| = k}.\\
        &= \Pr{s \sim t \mid s,t \in \Gamma(v)}\\
        &= \Pr{\Delta \mid v \sim s, t}.
    \end{align*} If $k = |\Gamma(v)| < 2$, we have that $\textsc{CC}_G(v) = 0$, implying that in total, \begin{align*}
        \Expected{\textsc{CC}(G)} = \Pr{\Delta \mid v \sim s, t}\Pr{\deg(v) \ge 2}
    \end{align*}
    as desired.

    For the second part, recall that we defined for all $i,j \in V$ the quantity $\kappa_{ij} = \min\{\lambda w_iw_j, n\}$ and note that 
    \begin{align*}
        \Pr{v \sim s, t} &= \frac{\min\{\lambda w_vw_s, n\}}{n} \frac{\min\{\lambda w_vw_t, n\}}{n} \\
        &\ge \frac{\min\{\lambda w_{\hat{v}}w_{\hat{s}}, n\}}{n} \frac{\min\{\lambda w_{\hat{v}}w_{\hat{t}}, n\}}{n}= \Pr{\hat{v} \sim \hat{s}, \hat{t}}
    \end{align*}
    because $\hat{v}$ is the vertex of minimal weight and because the events $\hat{v} \sim \hat{s}$ and $\hat{v} \sim \hat{t}$ are independent. Thus, \begin{align*}
        \Expected{\textsc{CC}_G(v)} &\le \Pr{\Delta \mid v \sim s, t} 
        = \frac{\Pr{\Delta}}{\Pr{v \sim s, t}} \le \frac{\Pr{\Delta}}{\Pr{\hat{v} \sim \hat{s}, \hat{t} }} = \Pr{\Delta \mid \hat{v} \sim \hat{s}, \hat{t} }.
    \end{align*}
\end{proof}

\subsection{$L_\infty$-norm}\label{sec:clusteringLinfty}

In this section, we analyse the clustering coefficient under $L_\infty$-norm, which results in \Cref{thm:clusteringLinfty}. To prove this theorem, we use \Cref{thm:sharpcliquebounds}.

\clusteringlinfty*
\begin{proof}
    (1) Upper Bounds.
    We use \Cref{lem:clusteringbound} and thus only need an upper bound on $\Expected{\textsc{CC}(G_{\le n^{1/8}})}$. For this, we use \Cref{lem:triangleclustering}, and we let $v, s, t$ be three random vertices in $G_{\le n^{1/8}}$ conditioned on the event that $v$ is of minimal weight among $v,s,t$. If we denote by $\Delta$ be the event that $v, s, t$ form a triangle, we get from \Cref{lem:triangleclustering} that $\Expected{\textsc{CC}(G_{\le n^{1/8}})} \le \Pr{\Delta \mid v \sim s,t}$. Accordingly, by \Cref{thm:sharpcliquebounds}, we may bound \begin{align*}
        \Expected{\textsc{CC}(G_{\le n^{1/8}})} &\le \Expected{\min \left\{ 1, \left . \frac{\max\{w_s, w_t\}}{w_v} \left( \frac{3}{4} \right)^d \right\} \right| v \sim s, t  }.
    \end{align*}
    To bound the expectation above, we analyze the distribution of $\max\{w_s, w_t\}$. Taking into account, that we consider $v,s,t \in G_{\le n^{1/8}}$ conditioned on the event that $w_v \le w_s, w_t$, a standard calculation shows that \begin{align*}
        &\Pr{\max \{ w_s, w_t \} \ge x \mid v \sim s, t } \\
        &\hspace{1cm}= (1 + o(1)) \left( 1 - \left( 1 - \left( \frac{x}{w_v} \right)^{2-\beta} \right)^2 \right) \le (2 + o(1)) \left( \frac{x}{w_v} \right)^{2-\beta}
    \end{align*} where the factor of $1 + o(1)$ comes from conditioning on $v, s, t \in G_{\le n^{1/8}}$ and the exponent of $2-\beta$ is due to the weight bias arising from conditioning on $v \sim s, t$. Therefore, if $\beta > 3$, the above random variable has finite expectation and we can bound \begin{align*}
        \Expected{\min \left\{ 1, \left . \frac{\max\{w_s, w_t\}}{w_v} \left( \frac{3}{4} \right)^d \right\} \right| v \sim s, t  } &\le \left( \frac{3}{4} \right)^d \Expected{\left . \frac{\max\{w_s, w_t\}}{w_v} \right| v \sim s, t }\\
        &\le (2 + o(1))\left( \frac{3}{4} \right)^d \int_{1}^\infty (\beta-2) x^{\beta-2} \d x\\
        &\le (2 + o(1)) \frac{\beta - 2}{\beta - 3} \left( \frac{3}{4} \right)^d = \mathcal{O}_d\left( \left(\frac{3}{4}\right)^d \right),
    \end{align*} which finishes the case $\beta > 3$. For the case $\beta < 3$, we instead bound \begin{align*}
        &\Expected{\min \left\{ 1, \left . \frac{\max\{w_s, w_t\}}{w_v} \left( \frac{3}{4} \right)^d \right\} \right| v \sim s, t  }\\
        &\hspace{1cm}\le (2 + o(1))\left( \frac{3}{4} \right)^d \int_{1}^{\left(\frac{4}{3}\right)^d} (\beta-2) x^{2-\beta} \d x + \Pr{\frac{\max\{w_s, w_t\}}{w_v} \ge \left(\frac{4}{3}\right)^d \mid v \sim s, t}\\
        &\hspace{1cm}\le (2 + o(1)) \frac{\beta - 2}{3 - \beta} \left( \frac{3}{4} \right)^d\left( \frac{4}{3} \right)^{d(3  - \beta )} + (2 + o(1)) \left( \frac{3}{4} \right)^{d(\beta - 2)}= \mathcal{O}_d \left( \left( \frac{3}{4} \right)^{(\beta - 2)d} \right)
    \end{align*} as desired.

    (2) Lower Bounds. First of all, we show that $\Expected{\textsc{CC}(G_{\le n^{1/8}})}= \Omega_d((3/4)^d)$. To this end, we let $v, s, t$ be three random vertices in $G_{\le n^{1/8}}$ and we let $\hat{v}, \hat{s}, \hat{t}$ be the vertices $v,s,t$ reordered by their weights such that the weight of $\hat{v}$ is minimal. Then,  \begin{align}\label{eq:weightchange}
        \Pr{\Delta \mid v \sim s, t} = \frac{\Pr{\Delta}}{\Pr{v \sim s,t}} = \frac{\Pr{\hat{v} \sim \hat{s},\hat{t}}}{\Pr{v \sim s,t}} \Pr{\Delta \mid \hat{v} \sim \hat{s},\hat{t}} = \frac{w_{\hat{v}}}{w_v} \Pr{\Delta \mid \hat{v} \sim \hat{s},\hat{t}}.
    \end{align} Furthermore, by \Cref{lem:triangleclustering}, we have \begin{align*}
        \Expected{\textsc{CC}(G_{\le n^{1/8}})} &= \Pr{\Delta \mid v \sim s, t} \Pr{\deg(v) \ge 2}\\
        &\ge \Pr{\Delta \cap (w_{\hat{s}}, w_{\hat{t}} \le c w_{\hat{v}}) \mid v \sim s, t} \Pr{\deg(v) \ge 2}\\
        &= \Pr{\Delta \mid (v \sim s, t) \cap (w_{\hat{s}}, w_{\hat{t}} \le c w_{\hat{v}})} \\
        &\hspace{2cm} \cdot \Pr{w_{\hat{s}}, w_{\hat{t}} \le c w_{\hat{v}} \mid v \sim s, t} \Pr{\deg(v) \ge 2}
    \end{align*} where $c$ is an arbitry constant greater than $1$ (that does not depend on $d$). Now it is easy to see that $\Pr{w_{\hat{s}}, w_{\hat{t}} \le c w_{\hat{v}} \mid v \sim s, t }$ and $\Pr{\deg(v) \ge 2}$ are both constant. To bound the remaining factor, we note that $w_{\hat{v}}/w_v \ge 1/c$ conditioned on $w_{\hat{s}}, w_{\hat{t}} \le c w_{\hat{v}}$, so we can use (\ref{eq:weightchange}) to obtain \begin{align*}
        \Pr{\Delta \mid (v \sim s, t) \cap (w_{\hat{s}}, w_{\hat{t}} \le c w_{\hat{v}})} &\ge \frac{1}{c} \Pr{\Delta \mid (\hat{v} \sim \hat{s}, \hat{t}) \cap (w_{\hat{s}}, w_{\hat{t}} \le c w_{\hat{v}})} \ge \frac{1}{c} \left(\frac{3}{4}\right)^d
    \end{align*} where in the last step, we used \Cref{thm:sharpcliquebounds} which is now applicable because $\hat{v}$ is of minimal weight among $v,s,t$. Together, this shows that $\Expected{\textsc{CC}(G_{\le n^{1/8}})} = \Omega_d((3/4)^d)$ as desired. 

    However, if $\beta$ is very close to $2$, we derive a better lower bound. To this end, we note that $\Pr{\Delta \mid v \sim s,t} = 1$ if $w_{s}, w_{t} \ge 2^d w_{v}$, which is easy to show using the respective connection thresholds. Hence, $\Expected{\textsc{CC}(G_{\le n^{1/8}})} \ge \text{Pr}[w_{s}, w_{t} \ge 2^d w_{v} \mid v \sim s, t] \Pr{\deg(v) \ge 2}$. Now, using that $w_{v}$ is at most a constant with constant probability and then applying similar calculations regarding the weight distribution of $s, t$ as in part (1) of this proof yields that \begin{align*}
        \Expected{\textsc{CC}(G_{\le n^{1/8}})} &\ge \text{Pr}[w_{s}, w_{t} \ge 2^d w_{v} \mid v \sim s, t] \Pr{\deg(v) \ge 2} \\
        &\ge C (2^{d(2 - \beta)})^2\\ 
        &= \Omega_d((1/4)^{(\beta - 2)d})
    \end{align*} for some constant $C > 0$, as desired.
 

    
\end{proof}

\subsection{General $L_p$-norms}\label{sec:clusteringLp}

In this section, we generalize the previous result to other $L_p$-norms for $1 \le p < \infty$. We show that, in the threshold model, one also obtains an upper bound on the clustering coefficient that decreases exponentially with $d$ and holds with high probability. Although we do not have an explicit bound for the base of this exponential function, this result illustrates that using a different norm does not drastically change the behavior of the clustering coefficient.

\clusteringlp*

We start with deriving probability theoretic methods for analyzing
random vectors uniformly distributed in the unit ball under $L_p$-norm
and afterwards use them to bound the clustering coefficient (\Cref{sec:clusteringLpbound}).

\subsubsection{Probability-theoretic methods}

We start by introducing the following useful property of the distribution of a random vector $\F{x} \in \mathbb{R}^d$, which will afterwards allow us view $\F{x} =\| \F{x} \|_p \frac{\F{x}}{\|\F{x}\|_p} $ where $\| \F{x} \|_p$ and $\frac{\F{x}}{\|\F{x}\|_p}$ are independent. In the following we show this formally and analyze the distribution of these random variables. We start with $\frac{\F{x}}{\|\F{x}\|_p}$ and define the following useful property of a random vector.

\begin{definition}[$L_p$-Symmetry]
    Let $\F{x}\in\mathbb{R}^d$ be a random vector with density function $\rho: \mathbb{R}^d \rightarrow \mathbb{R}_{\ge0}$. We refer to $\rho$ and $\F{x}$ as $L_p$-symmetric if for all $\F{y}, \F{z} \in \mathbb{R}^d$ with $\norm{\F{y}}_p = \norm{\F{z}}_p$, we have $\rho(\F{y}) = \rho(\F{z})$. As this implies that $\rho$ only depends on the norm $r \in \mathbb{R}$ of its argument, we also denote with $\rho(r)$ the value of $\rho$ for any $\F{z} \in \mathbb{R}^d$ with $\|\F{z}\|_p = r$.
\end{definition}

It is easy to see that $\F{x} \sim B_p(1)$ has the above property. We shall further see that any two $L_p$-symmetric random vectors $\F{y}, \F{y'}$ are equivalent in the sense that their "directions" $\F{y}/\|\F{y}\|_p$ and $\F{y'}/\|\F{y'}\|_p$ are identically distributed. This allows us to sample the random vector $\F{x}/\|\F{x}\|_p$ from an arbitrary $L_p$-symmetric distribution. 

\begin{lemma}[Equivalence of $L_p$-Symmetric Density Functions]\label[lemma]{lem:lpsymmetry}
    Let $\F{x}, \tilde{\F{x}} \in \mathbb{R}^d$ be two random vectors with density functions $\rho$ and $\tilde{\rho}$ respectively, both $L_p$-symmetric. Then, the random vectors $\F{z} \coloneqq \frac{\F{x}}{\norm{\F{x}}_p}$ and $ \tilde{\F{z}} \coloneqq \frac{\tilde{\F{x}}}{\norm{\tilde{\F{x}}}_p}$ are identically distributed.
\end{lemma} 

Before we prove this lemma, we introduce some further notation and some auxiliary statements. Let $S \subseteq S_p(1)$ be some subset of the (surface of the) unit sphere under $L_p$-norm. We define the set $S(r) = \{ \F{x} \in \mathbb{R}^d \mid \F{x}/\|\F{x}\|_p \in S, \|\F{x}\|_p \le r \}$, which contains all vectors from $\mathbb{R}^d$ with norm at most $r$ that are in $S$ when projected to $S_p(1)$. We further denote by $\vol(r)$ the volume of the unit ball of radius $r$ and by $\vol_S(r)$ the volume of the set $S(r)$. 
We start by showing the following useful property. \begin{lemma}\label[lemma]{lem:volume}
    Let $S \subseteq S_p(1)$, and let $S(r)$, $\vol_S(r)$, and $\vol(r)$ be defined as above. We have \begin{align*}
        \vol_S(r) = r^d \cdot \vol_S(1) = \vol(r) \frac{\vol_S(1)}{\vol(1)}.
    \end{align*}
\end{lemma} \begin{proof}
    We note that for any $r \ge 0$, \begin{align*}
        \vol_S(r) &= \int_{-\infty}^\infty \ldots \int_{-\infty}^\infty \mathds{1}((x_1,\ldots,x_d) \in S(r)) \d x_1\ldots \d x_d.
    \end{align*}
    Substituting $x_i = r \cdot y_i$ yields
    \begin{align*}
        \vol_S(r) &= \int_{-\infty}^\infty \ldots \int_{-\infty}^\infty \mathds{1}(r\cdot(y_1,\ldots,y_d) \in S(r)) r^d \d y_1\ldots \d y_d\\
        &= r^d \int_{-\infty}^\infty \ldots \int_{-\infty}^\infty \mathds{1}((y_1,\ldots,y_d) \in S(1)) \d y_1\ldots \d y_d\\
        &= r^d \vol_S(1).
    \end{align*}
    which shows the first part of our statement. For the second part, we observe that $\vol(r) = \vol_{S_p(1)}(r)$, and thus immediately obtain $\vol(r) = r^d \vol(1).$ 
    Hence, $r^d = \vol(r)/\vol(1)$, implying that $\vol_S(r) = \vol(r) \frac{\vol_S(1)}{\vol(1)}$.
\end{proof}
We continue by showing that we can express the probability of the event $\F{x}/\|\F{x}\|_p \in S$ for any $L_p$-symmetric random vector $\F{x}$ in the following way.
\begin{lemma}\label[lemma]{lem:integraldv}
    Let $\F{x} \in \mathbb{R}^d$ be a random vector with $L_p$-symmetric density function $\rho$ and let $S \subseteq S_p(1)$. We have \begin{align*}
        \Pr{\frac{\F{x}}{\|\F{x}\|_p} \in S} = \int_0^\infty \rho(r)\frac{\d \vol_S(r)}{\d r} \d r. 
    \end{align*}
\end{lemma}
\begin{proof}
    We define for any $\F{x} \in \mathbb{R}^d$ the indicator function \begin{align*}
        \mathds{1}_S(\F{x}) = \begin{cases}
            1 & \text{if } \F{x}/\|\F{x}\|_p \in S\\
            0 & \text{otherwise.}
        \end{cases}
    \end{align*} Furthermore, we define $\F{z} \coloneqq \F{x}/\|\F{x}\|_p$. For simplicity, we assume that $S$ is located in only one of the $2^d$ orthants of the standard $d$-dimensional cartesian coordinate system, the argumentation for the case where $S$ spans multiple orthants are analogously obtained by splitting $S$ into parts that each span one orthant, and afterwards summing over them. Therefore, in the following, we assume that $S \subseteq \mathbb{R}_{>0}^d$. We note that we may express \begin{align} \label{eq:integraly}
        \Pr{\F{z} \in S} &= \int_{\mathbb{R}_{>0}^d}\mathds{1}_S(\mathbf{x}) \rho(\mathbf{x}) \d \mathbf{x}.
    \end{align} where $\mathbf{x} = (x_1, \ldots, x_d)^T$.  We get from \cite[Theorem 3-13, page 67]{Spivak_1998} that if $A \subset \mathbb{R}^d$ is an open set and if $\varphi: A \rightarrow \mathbb{R}^d$ is an injective, continuously differentiable function such that $\det(J\varphi(\F{x})) \neq 0$ for all $\F{x} \in A$, then if $f : \varphi(A)  \rightarrow \mathbb{R}^d $ is integrable, \begin{equation*}
        \int_{\varphi(A)} f(\F{x}) \d \F{x} = \int_{A} f(\varphi(\F{y}))|\det(J\varphi(\F{y}))|\d \F{y}, 
    \end{equation*} where $J\varphi(\F{x})$ denotes the Jacobian matrix of $\varphi$ at the point $\F{x}$. We define $A_r$ as the open set $A_r = \{ (r, x_2, \ldots, x_d) \in \mathbb{R}_{>0}^d \mid \sum_{i=2}^d x_i^p < r^p \}$ and $A = \bigcup_{r>0} A_r$. Furthermore, we let \begin{align*}
        \varphi: A \rightarrow \mathbb{R}^d, (r, x_2, \ldots, x_d) \mapsto \left( \left(r^p - \sum_{i=2}^dx_i^p \right)^{1/p}, x_2, \ldots, x_d \right).
    \end{align*} 
    We note that this function is injective and that it has the remarkable property that for any $\F{x} = (r, x_2, \ldots, x_d) \in A$, $\norm{\varphi(\F{x})}_p = r$. Furthermore, we have $J\varphi_{ij} = 0$ for $i,j \ge 2, i \neq j$, $J\varphi_{ij} = 1$ for $i = j \ge 2$ and \begin{align*}
        J\varphi_{11} = \frac{\partial}{\partial r} \left(r^p - \sum_{i=2}^dx_i^p \right)^{1/p} = r^{p-1}\left(r^p - \sum_{i=2}^dx_i^p \right)^{1/p-1}.
    \end{align*} Furthermore, for all $i \ge 2$, we have \begin{align*}
        J\varphi_{1i} = \frac{\partial}{\partial x_i} \left(r^p - \sum_{i=2}^dx_i^p \right)^{1/p} = -x_i^{p-1}\left(r^p - \sum_{i=2}^dx_i^p \right)^{1/p-1}.
    \end{align*} Hence, $\varphi$ is continuously differentiable. Moreover, since $A \subseteq \mathbb{R}_{>0}^d$, we get that for all $1 \le i \le d$ and $\F{x} \in A$, we have $J\varphi_{1i} \neq 0$ and $J\varphi_{ii} \neq 0$, but for all $i,j \ge 2, i\neq j$, we have $J\varphi_{ij} = 0$. For this reason the columns of $J\varphi(\F{x})$ are not linearly dependent and so $\det(d\varphi(\F{x})) \neq 0$. In the following, we denote $|\det(J\varphi(\F{x}))|$ with $g(\F{x})$. We can hence transform \Cref{eq:integraly} as \begin{align*}
        \Pr{\F{z} \in S} &= \int_{\mathbb{R}_{>0}^d}\mathds{1}_S(\mathbf{x}) \rho(\mathbf{x}) \d \mathbf{x}\\ 
        &= \int_A \mathds{1}_S(\varphi(\mathbf{y})) \rho(\varphi(\mathbf{y})) g(\mathbf{y})\d \mathbf{y} \\ 
        &= \int_0^\infty \ldots \int_0^\infty \mathds{1}_S(\varphi(\mathbf{y})) \mathds{1}(\mathbf{y}\in A) \rho(\varphi(\mathbf{y})) g(\mathbf{y}) \d x_d \ldots \d x_2 \d r,
    \end{align*} where $\mathbf{y} = (r, x_2,\ldots, x_d)$ and $\mathds{1}(\mathbf{y}\in A)$ is an indicator function, which is equal to $1$ if $\mathbf{y}\in A$ and $0$ otherwise. We note that for any $\mathbf{y} = (r, x_2,\ldots, x_d) \in A$, we have $\norm{\varphi(\mathbf{y})}_p = r$. Since $\rho(\F{x})$ is $L_p$-symmetric it only depends on the norm of $\F{x}$, hence $\rho (\varphi(\mathbf{y})  )$ only depends on the first component $r$ of $\mathbf{y}$. We may therefore rewrite $\rho (\varphi(\mathbf{y})) = \rho(r)$ and rearrange \begin{align*}
        \Pr{\F{z} \in S} &= \int_0^\infty \rho(r) \int_0^\infty \ldots \int_0^\infty \mathds{1}_S(\varphi(\mathbf{y})) \mathds{1}(\mathbf{y}\in A) g(\mathbf{y}) \d x_d \ldots \d x_2 \d r.
    \end{align*} We define for any $r > 0$, \begin{align*}
        v_S(r) \coloneqq \int_0^\infty \ldots \int_0^\infty \mathds{1}_S(\varphi(\mathbf{y})) \mathds{1}(\mathbf{y}\in A) g(\mathbf{y}) \d x_d \ldots \d x_2
    \end{align*} and thus obtain 
    \begin{equation}\label{eq:volumer}
        \Pr{\F{z} \in S} = \int_0^\infty \rho(r) v_S(r) \d r.
    \end{equation} Now, recall that $\vol_S(R)$ is the volume of the set $S(R) = \{x \in \mathbb{R}^d \mid x/\norm{x}_p \in S, \norm{x}_p \le R\}$. We show that in fact $v_S(R) = \frac{\d \vol_S(R)}{\d R}$ for all $R > 0$. This gives \Cref{eq:volumer} an intuitive interpretation as integrating $\rho$ over $r$ along the sphere radius $r$ under $L_p$-norm. Note that \begin{align*}
    \vol_S(R) = \int_{\mathbb{R}^d} \mathds{1}(\mathbf{x} \in S(R)) \d\mathbf{x}.
    \end{align*}
    Now, with the same argumentation as above (and by omitting $\rho$), we obtain \begin{align*}
        \vol_S(R) &= \int_0^\infty \ldots \int_0^\infty \mathds{1}(\varphi(\mathbf{y}) \in S(R)) \mathds{1}(\mathbf{y}\in A) g(\mathbf{y}) \d x_d \ldots \d x_2 \d r\\
        &= \int_0^\infty \mathds{1}(r \le R) \int_0^\infty \ldots \int_0^\infty \mathds{1}_S(\varphi(\mathbf{y})) \mathds{1}(\mathbf{y}\in A) g(\mathbf{y}) \d x_d \ldots \d x_2 \d r\\ 
        &= \int_0^\infty \mathds{1}(r \le R) v_S(r) \d r = \int_0^R v_S(r) \d r
    \end{align*} where we used that for all $\mathbf{y} \in A$, we have $\mathds{1}(\varphi(\mathbf{y}) \in S(R)) = \mathds{1}(r \le R)\mathds{1}_S(\varphi(\mathbf{y}))$. Applying the Leibnitz integral rule, we get $\frac{\d \vol_S(R)}{\d R} = v_S(R)$, which finishes the proof.
\end{proof}

The above two statements imply the following corollary, which in turn implies \Cref{lem:lpsymmetry}. \begin{corollary}\label{cor:volumeprobability}
    Let $\F{x}$ be an $L_p$-symmetric random vector and let $S \subseteq S_p(1)$. We have \begin{align*}
        \Pr{\frac{\F{x}}{\|\F{x}\|_p} \in S} = \frac{\vol_S(1)}{\vol(1)}.
    \end{align*}
\end{corollary} \begin{proof}
    define $\F{z} \coloneqq \F{x}/\|\F{x}\|_p$. By \Cref{lem:integraldv}, we may express \begin{align*}
        \Pr{\F{z} \in S} = \int_0^\infty \rho(r)\frac{\d \vol_S(r)}{\d r} \d r.
    \end{align*} 
    Furthermore, we have by \Cref{lem:volume} that $\vol_S(R) = \vol(R)\frac{\vol_S(1)}{\vol(1)}$ and hence, \begin{align*}
        \frac{\d \vol_S(R)}{\d R} = \frac{\vol_S(1)}{\vol(1)} \frac{\d \vol(R)}{\d R}.
    \end{align*} 
    Accordingly, \begin{align*}
        \Pr{\F{z} \in S} &= \int_{0}^\infty \rho(r) \cdot \frac{\d \vol_S(R)}{\d r} \d r \\
        &= \frac{\vol_S(1)}{\vol(1)} \int_{0}^\infty \rho(r) \cdot \frac{\d \vol(R)}{\d r} \d r.
    \end{align*}
    We note that $\Pr{\F{z} \in S_p(1)} = 1$, and so, by \Cref{lem:integraldv}, we get
    \begin{align*}
        \int_{0}^\infty \rho(r) \cdot \frac{\d \vol(R)}{\d r}\d r = \Pr{\F{z} \in S_p(1)} = 1.
    \end{align*} This shows \begin{align*}
        \Pr{\F{z} \in S} = \frac{\vol_S(1)}{\vol(1)}.
    \end{align*}
\end{proof}

With this statement, we may now prove \Cref{lem:lpsymmetry}.
\begin{proof}[Proof of \Cref{lem:lpsymmetry}.]
We show that for any $S \subseteq S_p(1)$, we have that $\Pr{\F{z} \in S} = \Pr{\tilde{\F{z}} \in S}$. Because $\F{x}$ and $\tilde{\F{x}}$ are both $L_p$-symmetric, we get by \Cref{cor:volumeprobability} that both $\Pr{\F{z} \in S}$ and $\Pr{\tilde{\F{z}} \in S}$ are equal to $\frac{\vol_S(1)}{\vol(1)}$, which directly implies the desired statement.
\end{proof}

\paragraph*{The $\rchi^p$-Distribution} 
In addition to the distribution of $\F{x} \sim B_p(1)$, we need another $L_p$-symmetric distribution. For this purpose recall the definitions of the $\rchi_p(d)$ and the $\rchi^p(d)$ distributions from the introduction.
It is easy to see that a random vector $\F{x} \sim \rchi_p(d)$ is $L_p$-symmetric by observing that its density function is \begin{equation*}
    \rho_{\F{x}}(\F{x}) = \prod_{i=1}^d \gamma e^{-\frac{1}{2}|\F{x}(i)|^p} = \gamma^d e^{-\frac{1}{2}\sum_{i=1}^d|\F{x}(i)|^p} = \gamma^d e^{-\frac{1}{2}(\|\F{x}\|_p)^p}
\end{equation*} and thus only depends on the norm of $\F{x}$. We further note that for the case $p=2$, $\rchi_2(d)$ is the standard $d$-variate normal distribution $\mathcal{N}(0, I_d)$ (where $I_d$ is the $d \times d$ identity matrix), and that $\rchi^2(d)$ is the chi-squared distribution with $d$ degrees of freedom. The distribution $\rchi^p(d)$ can hence be seen as a generalization of the chi-squared distribution to other $L_p$-norms.

We further verify that $\gamma$ is indeed the correct normalization constant. For this, let $X \sim \rchi_p(1)$ and observe that \begin{align*}
    1 = \int_{-\infty}^\infty \rho_X(x)\d x = \gamma \cdot 2\int_0^\infty e^{-\frac{1}{2}x^p}\d x.
\end{align*} With the substitution $x = (2y)^{\frac{1}{p}}$, we obtain \begin{align*}
    &\gamma \cdot 2\int_0^\infty e^{-\frac{1}{2}x^p}\d x
    = \gamma \cdot 2 \int_0^\infty \frac{2^{1/p}}{p} y^{1/p - 1} e^{-y}\d y
    =\gamma \frac{2^{1/p+1} \Gamma\left(\frac{1}{p}\right)}{p}.
\end{align*}
We hence get \begin{equation*}
    \gamma = \frac{p}{2^{\frac{1}{p}+1}\Gamma\left(\frac{1}{p}\right)},
\end{equation*}
Note that for $p=2$, one does indeed obtain the correct normalization constant of the standard normal distribution $\mathcal{N}(0,1)$, which is equal to $1/\sqrt{2\pi}$.

We continue with deriving a tail bound on the $\rchi^p(d)$ distribution and start with deriving its moment-generating function. \begin{lemma}\label[lemma]{lem:moment-gen-func}
    Let $Z \sim \rchi^p(1)$. Let $\psi_Z$ be the moment generating function of $Z$, defined as \begin{equation*}
        \psi_{Z}: \mathbb{R}_0^+ \rightarrow \mathbb{R}, \psi_{Z}(\lambda) = \Expected{e^{\lambda Z}}.
    \end{equation*} Then, for every $\lambda < \frac{1}{2}$, we have 
    \begin{align*}
         \psi_Z(\lambda) = \left( 1 - 2\lambda \right)^{-\frac{1}{p}}.
    \end{align*}
\end{lemma} \begin{proof} Let $X \sim \rchi_p(1)$ and note that we may write $Z = |X|^p$.
    Recall that the probability density of $X$ is $\rho_X(x) = \gamma e^{-\frac{1}{2}|x|^p}$.
    Denote by $\rho_Z$ the density function of $Z$ and observe that \begin{align*}
        \rho_Z(x) &= \frac{\d \Pr{Z \ge x}}{\d x} 
                = \frac{\d \Pr{|X|^p \ge x}}{\d x}
                = \frac{\d \Pr{|X| \ge x^\frac{1}{p}}}{\d x}
                = \rho_{|X|}\left(x^{\frac{1}{p}}\right) \frac{\d x^\frac{1}{p}}{\d x}\\
               &= 2 \rho_X\left(x^{\frac{1}{p}}\right) \frac{\d x^\frac{1}{p}}{\d x} = 2\gamma e^{-\frac{1}{2}x} \frac{1}{p} x^{\frac{1}{p}-1}
                = \frac{x^{\frac{1}{p}-1}e^{-\frac{1}{2}x}}{2^{\frac{1}{p}} \Gamma\left(\frac{1}{p}\right) }.
    \end{align*} Note that, in the fifth equality, we used that $\rho_{|X|}(x) = 2\rho_{X}(x)$.
    We continue by deriving the moment-generating function of the random variable~$Z$. We obtain \begin{align*}
        \psi_Z(\lambda) &= \Expected{e^{\lambda Z}} = \int_0^\infty \rho_Z(x) e^{\lambda x} dx \\
                        &= \frac{1}{2^{\frac{1}{p}} \Gamma\left(1/p\right)} \int_0^\infty x^{\frac{1}{p}-1}e^{-x(1/2-\lambda)} dx.
    \end{align*} We note that this integral exists for $\lambda < \frac{1}{2}$. With the substitution $x = y (1/2-\lambda)^{-1}$, it transforms to \begin{align*}
        \psi_Z(\lambda) &= \frac{1}{2^{\frac{1}{p}} \Gamma\left(1/p\right)} \int_0^\infty x^{\frac{1}{p}-1}e^{-x(1/2-\lambda)} dx\\
                        &= \frac{1}{2^{\frac{1}{p}} \Gamma\left(1/p\right)} \int_0^\infty y^{\frac{1}{p}-1}e^{-y} \frac{(1/2 - \lambda)^{1 - \frac{1}{p}}}{1/2 - \lambda} dy \\
                        &= \frac{(1/2 - \lambda)^{-\frac{1}{p}}}{2^{\frac{1}{p}} \Gamma\left(1/p \right)}  \Gamma \left( 1/p \right)\\
                        &= \left( 1 - 2\lambda \right)^{-\frac{1}{p}}.
    \end{align*}
\end{proof}

\begin{corollary}\label{cor:expectationchip}
    Let $Z \sim \rchi^p(d)$. Then, \begin{align*}
        \Expected{Z} = \frac{2d}{p}.
    \end{align*}
\end{corollary} \begin{proof}
    Let $X \sim \rchi^p(1)$. We get $\Expected{Z} = d \cdot \Expected{X}$ as $Z$ is the sum of $d$ independent random variables distributed identically as $X$. We further note that the expectation of $X$ is equal to the derivative of its moment-generating function at $\lambda = 0$. We get from \Cref{lem:moment-gen-func} that \begin{align*}
        \frac{\d \psi_X(\lambda)}{\d \lambda} = \frac{2}{p} (1 - 2\lambda)^{-\frac{1}{p}-1}
    \end{align*} and hence, $\Expected{X} = \frac{2}{p}$.
\end{proof}

We continue by showing that a random variable $Z \sim \rchi^p(d)$ is concentrated around its expected value. Under the hood, our bounds are obtained in the same way as the Chernoff-Hoeffding bounds, namely by applying Markov's inequality to the moment generating function of $Z$. However, instead of doing this directly, we take a shortcut by applying the following variant of Bernstein's inequality that is proven by Massart in \cite{Morel_Massart_Picard_Takens_Teissier_2007}.

\begin{theorem}[Proposition 2.9 in \cite{Morel_Massart_Picard_Takens_Teissier_2007}]\label{thm:bernstein} Let $X_1, \ldots, X_d$ be independent, real-valued random variables. Assume that there exist constants $v, c > 0$ such that \begin{equation*}
    \sum_{i=1}^d \Expected{X_i^2} \le v
\end{equation*}
and that for all integers $k\ge 3$, \begin{equation*}
    \sum_{i=1}^d \Expected{\abs{X_i}^k} \le \frac{k!}{2} v c^{k-2}.
\end{equation*} 
Let $S = \sum_{i=1}^d (X_i - \Expected{X_i})$. Then, for every $x > 0$, \begin{equation*}
    \Pr{S \ge \sqrt{2vx} + cx} \le \exp(-x).
\end{equation*}
\end{theorem}
With this, we are able to show the following. 
\begin{theorem}\label{thm:p-bound}
    Let $X_1, \ldots, X_d$ be i.i.d. random variables from $\rchi_p(1)$ and define the random variable $Z \coloneqq \sum_{i=1}^d|X_i|^p$. Note that $Z \sim \rchi^p(d)$. Then, for all $x > 0$,
    \begin{enumerate}
        \item[(i)] $\Pr{Z \ge \Expected{Z} + 2\sqrt{2\Expected{Z}x} + 2x } \le \exp(-x)$
        \item[(ii)] $\Pr{Z \le \Expected{Z} - 2\sqrt{2\Expected{Z}x} - 2x } \le \exp(-x)$.
    \end{enumerate}
\end{theorem}
\begin{proof}
    We use \Cref{thm:bernstein}. To show that the random variables $|X_1|^p, \ldots, |X_d|^p$ fulfill the conditions of \Cref{thm:bernstein}, we derive bounds on its moments. For any $X \sim \rchi_p(1)$, define $Y = |X|^p$. We use the moment generating function from \Cref{lem:moment-gen-func} to derive bounds on the moments of $Y$. For all integers $k \ge 0$, we note that we have $\Expected{Y^{k}} = \psi_Y^{(k)}(0)$, where $\psi_Y^{(k)}$ denotes the $k$-th derivative of $\psi_Y$. We note that \begin{equation*}
        \psi_Y'(\lambda) = \frac{2}{p} \left( 1 - 2\lambda \right)^{-\frac{1}{p}-1}
    \end{equation*} and 
    \begin{equation*}
        \psi_Y''(\lambda) = \frac{4}{p}\left(\frac{1}{p}+1\right) \left( 1 - 2\lambda \right)^{-\frac{1}{p}-2},
    \end{equation*}
    from which we derive $\Expected{Y} = \frac{2}{p}$ and $\Expected{Y^2} = \frac{4}{p}\left(\frac{1}{p}+1\right)$. For $k \ge 3$ one can easily verify that \begin{align*}
        \psi_Y^{(k)}(\lambda) &= \left( 1 - 2\lambda \right)^{-\frac{1}{p}-k} \Expected{Y^2} 2^{k-2} \prod_{i=2}^{k-1}\left( \frac{1}{p} + i \right)
    \end{align*} and hence, \begin{align}\label{eq:momentbound}
        \Expected{Y^k} = \psi_Y^{(k)}(0) &= \Expected{Y^2} 2^{k-2} \prod_{i=2}^{k-1}\left( \frac{1}{p} + i \right) \nonumber \\
        &= \Expected{Y^2} 2^{k-2} \prod_{i=1}^{k-2}\left( \frac{1}{p} + i + 1 \right) \nonumber \\
        &\le \Expected{Y^2} 2^{k-2} \prod_{i=1}^{k-2}\left( i + 2 \right) = \Expected{Y^2} 2^{k-2} \frac{k!}{3!} \le \Expected{Y^2} 2^{k-1} \frac{k!}{2}. 
    \end{align}
     Recall that we have $\Expected{Y^2} = \frac{4}{p}\left(\frac{1}{p}+1\right)$ and hence, $\Expected{Y^2} \le \frac{8}{p}$ due to $p \ge 1$. If we define $Y_i = |X_i|^p$ and set $v = 8d/p, c = 2$, we have that \begin{equation*}
         \sum_{i=1}^d \Expected{Y_i^2} \le \frac{8d}{p} = v
     \end{equation*} and thus, for all $k \ge 3$, \begin{equation*}
        \sum_{i=1}^d \Expected{Y_i^k} \le d\Expected{Y^2} 2^{k-1} \frac{k!}{2} \le \frac{k!}{2}vc^{k-2},
     \end{equation*} which shows that the conditions of \Cref{thm:bernstein} are fulfilled. Since $Z = \sum_{i=1}^d Y_i$ and $\Expected{Z} = \frac{2d}{p}$, we get that for all $x > 0$, \begin{align*}
        \Pr{Z - \Expected{Z} \ge \sqrt{16d/p \cdot x} + 2x} &= \Pr{Z \ge \Expected{Z} + 2\sqrt{2\Expected{Z} \cdot x} + 2x}\\ &\le \exp(-x),
    \end{align*} which shows the first statement. 

    For the second statement, we define $Y_i' \coloneqq -Y_i$ and note that $-Z = \sum_{i=1}^dY_i'$. Furthermore, we have that $\Expected{Y_i'^2} = \Expected{Y_i^2}$ and $\Expected{\abs{Y_i'}^k} = \Expected{Y_i^k}$ for all integers $k \ge 0$. We have that \begin{equation*}
        \sum_{i=1}^d \Expected{Y_i'^2} = \sum_{i=1}^d \Expected{Y_i^2} \le \frac{8d}{p} = v
    \end{equation*} and for all $k\ge 3$, we get from \Cref{eq:momentbound} that \begin{equation*}
        \sum_{i=1}^d \Expected{\abs{Y_i'}^k} = \sum_{i=1}^d \Expected{Y_i^k} \le \frac{k!}{2}vc^{k-2}.
    \end{equation*}
    Hence, it follows from \Cref{thm:bernstein} that \begin{align*}
        \Pr{-Z + \Expected{Z} \ge \sqrt{16d/p \cdot x} + 2x} &= \Pr{Z \le \Expected{Z} - 2\sqrt{2\Expected{Z}x} - 2x} \\
        &\le \exp(-x),
    \end{align*}
    which implies the second statement.
\end{proof}

We can slightly reformulate this bound such that it is more convenient to work with them. Observe the similarity of the following bounds with the Chernoff-Hoeffding bound from \Cref{thm:chernoff-hoeffding}.

{\renewcommand{\thetheorem}{\ref{cor:tailchisquared}}
\tailchisquared*
}
\renewcommand{\thetheorem}{25}
\begin{proof}
    We use \Cref{thm:p-bound} and set $x = \delta \Expected{Z}$. We then obtain \begin{align*}
        \Pr{Z \ge \Expected{Z} + \Expected{Z} \cdot 2\sqrt{2\delta} + \Expected{Z} \cdot 2\delta} &= \Pr{Z \ge \Expected{Z} (1 + 2 \sqrt{2\delta} + 2\delta)}\\
        &\le \exp(-\delta \Expected{Z}).
    \end{align*} 
    Recalling from \Cref{cor:expectationchip} that $\Expected{Z} = \frac{2d}{p}$ then implies that $\Pr{Z \ge (1 + \varepsilon)\Expected{Z} } \le \exp(-2d\delta / p)$ for $\varepsilon = 2(\sqrt{2\delta}  + \delta)$. The argumentation for the second statement is analogous.
\end{proof}

\subsubsection{Bounding the clustering coefficient}\label[section]{sec:clusteringLpbound}

We use the insights gained so far to prove a bound on the probability that two random neighbors of a vertex $s$ that have bounded weight are adjacent.

\begin{lemma}\label[lemma]{lem:cluseringLpbound}
    Let $G = G(n, d, \beta, w_0)$ be a GIRG sampled under $L_p$-norm. After sampling the weights, let $s, u, v$ be three vertices in $G_{\le n^{1/8}}$ with $w_s \le w_u, w_v$ and $w_u, w_v \le c^d w_s$ for some $c>1$. Let $\Delta$ be the event that $\{s, u, v\}$ form a triangle. 
    Then, there exists a choice for $c$ such that there are constants $a, b > 0, c > 1$ such that for sufficiently large $n$ and all $d \ge 1, d = o(\log(n))$, \begin{align*}
        \Pr{\Delta\mid s \sim u, v} \le a \cdot \exp(-b d).
    \end{align*} Here, the randomness originates from the assignment of coordiates in $\mathbb{T}^d$ to $s,u,v$.
\end{lemma}
\begin{proof}
    Recall that $B_p(r)$ is the ball of radius $r$ under $L_p$ norm. We assume that $n$ is large enough such that the ball of volume $\lambda w_s^2c^{2d}/n$ has a radius of $r \le 1/4$. Note that this is possible since $d = o(\log(n))$. With this we may simply measure the distance of two points $\F{x},\F{y} \in B_p(r)$ as $\| \F{x} - \F{y}\|_p$ and assume that $t_{uv}$ is precisely the radius of the ball of volume $\lambda w_uw_v/n$.
    
    Now, assuming $s \sim u,v$ and $w_u,w_v \le c^d w_s$, the vertices $u,v$ are uniformly distributed within the balls $B_p(t_{sv})$ and $B_p(t_{su})$ (centered at the position of $s$), respectively. Assuming the position of $s$ is the origin of our coordinate system, we denote by $\F{x}_u, \F{x}_v$ the (random) positions of $u,v$. Hence, the probability that $u$ and $v$ are connected is simply $\text{Pr}[\norm{\F{x}_u - \F{x}_v}_p \le t_{uv}]$. If we denote by $\vol(r)$ the volume of the ball $B_p(r)$, we further note that $\vol(r) = r^d \vol(1)$ (cf. \Cref{lem:volume}), and since we choose $t_{uv}$ such that $\vol(t_{uv}) =  \lambda w_uw_v/n$, we get \begin{equation}\label{eq:thresholdlp}
        t_{uv} = \left( \frac{\lambda w_uw_v}{\vol(1)n} \right)^{1/d}.
    \end{equation}
    
    In the following, we derive an upper bound for $\Pr{\norm{\F{x}_u - \F{x}_v}_p \le t_{uv}}$. We note that we can equivalently describe the random variables $\F{x}_u, \F{x}_v$ as $\F{x}_u = t_{us} \F{y}_u$ and $\F{x}_v = t_{vs} \F{y}_v$, where $\F{y}_u$ and $\F{y}_v$ are i.i.d. random vectors uniformly distributed according to the standard Lebesgue measure in $B_p(1)$. With this, we reformulate the probability $\Pr{ \norm{ \F{x}_u - \F{x}_v }_p \le t_{uv} }$ as \begin{align*}
        \Pr{\norm{\F{x}_u - \F{x}_v}_p \le t_{uv}} &= \Pr{\norm{t_{us} \F{y}_u - t_{vs} \F{y}_v}_p \le t_{uv}}\\
        &= \Pr{\norm{\F{y}_u - (t_{vs}/t_{us})  \F{y}_v}_p \le t_{uv}/t_{us}}\\
        &= \Pr{\|\F{y}_u - \left(w_v/w_u\right)^{1/d} \F{y}_v\|_p \le \left(w_v/w_s\right)^{1/d}}.
    \end{align*}
    To find an upper bound for this probability, we instead lower bound the probability of the event that 
    \begin{equation*} 
        \norm{\F{y}_u - \left(w_v/w_u\right)^{1/d} \F{y}_v}_p > \left(w_v/w_s\right)^{1/d}.
    \end{equation*} 
    Since $w_v,w_s \in [w_s, c^d w_s]$, we have $\left(w_v/w_s\right)^{1/d} \le c$ and hence, it suffices to lower bound \begin{equation*}
        \Pr{\norm{\F{y}_u - \left(w_v/w_u\right)^{1/d} \F{y}_v}_p > c}
    \end{equation*} or equivalently \begin{equation*}\label{eq:radialdist}
        \Pr{\left( \norm{\F{y}_u - \left(w_v/w_u\right)^{1/d} \F{y}_v}_p \right)^p > c^p}.
    \end{equation*}
    For this, we start by investigating the properties of the random vectors $\F{y}_u, \F{y}_v \sim B_p(1)$. Recall from \Cref{lem:lpsymmetry} that we may equivalently express the random vector $\F{y} \sim B_p(1)$ as $\F{y} = \|\F{y}\|_p \cdot \F{y}/\|\F{y}\|_p$ where $\|\F{y}\|_p$ and $\F{y}/\|\F{y}\|_p$ are independent. Accordingly, $\F{y}$ is identically distributed as the product of a random variable $r$ identically distributed as $\|\F{y}\|_p$, and a random vector $\F{z}$ identically distributed as $\F{y}/\|\F{y}\|_p$.

    We note that $r$ and $\|\F{y}\|_p$ are distributed such that for any $0\le \zeta \le 1$, we have \begin{align*}
        \Pr{\norm{\F{y}}_p \le \zeta} = \frac{\vol_p(\zeta)}{\vol_p(1)} = \zeta^d
    \end{align*} and thus, \begin{align*}
        \Pr{\norm{\F{y}}_p \ge \zeta} = 1 - \zeta^d.
    \end{align*} Furthermore, due to the $L_p$-symmetry of $\F{y_u}, \F{y_v}$ and \Cref{lem:lpsymmetry}, we assume that $\F{z} = \F{\tilde{z}} / \|\F{\tilde{z}}\|_p$ where $\F{\tilde{z}}$ is a random vector from the $\rchi_p(d)$-distibution. 
 
    In the following, we hence assume that $\F{y}_u = r_u \cdot \F{\tilde{z}}_u/||\F{\tilde{z}}_u||_p$, and $\F{y}_v = r_v \cdot \F{\tilde{z}}_v/||\F{\tilde{z}}_v||_p$, for suitable, independent random variables $r_u, r_v$ and independent random vectors $\F{\tilde{z}}_u, \F{\tilde{z}}_v \sim \rchi_p(d)$.

    With this observation, we find a lower bound for \begin{equation*}
        \Pr{\left(\norm{\F{y}_u - \left(w_v/w_u\right)^{1/d} \F{y}_v}_p \right)^p > c^p}. 
    \end{equation*}
    We first rewrite the term $\left(\norm{\F{y}_u - \left(w_v/w_u\right)^{1/d} \F{y}_v}_p\right)^p$ as
    \begin{align*}
        \left(\norm{\F{y}_u - \left(\frac{w_v}{w_u}\right)^{1/d} \F{y}_v}_p\right)^p = \sum_{i=1}^d \abs{\F{y}_{u}(i) - \left(\frac{w_v}{w_u}\right)^{1/d} \F{y}_{v}(i)}^p = S_1 + S_2,
    \end{align*} where $S_1$ is the sum of all components in which $y_{ui}$ and $y_{vi}$ have opposite sign, and $S_2$ is the sum of all remaining components. We show that there are constants $a, b > 0, c > 1$ such that $S_1 + S_2$ is greater than $c^p$ with probability at least $1 - a \cdot \exp(-bd)$. In this section, we refer to an event as happening \emph{with overwhelming probability}\footnote{Note that this is a stricter notion of what is commonly referred to as ``with overwhelming probability'' in literature.} if there are constants $a, b > 0$ such that the event happens with probability at least $1 - a \cdot \exp(-b d)$. Note that, if two events $\Event{1}$ and $\Event{2}$ happen with overwhelming probability, then also $\Event{1} \cap \Event{2}$ happens with overwhelming probability as, by a union bound, we have $\Pr{\overline{\Event{1} \cap \Event{2}}} \le a \cdot \exp(-bd) + a' \cdot \exp(-b'd)$ for some $a, a', b, b' > 0$ and thus $\Pr{\Event{1} \cap \Event{2}} \ge 1 - 2\max\{a, a'\}\exp(-\max\{b, b'\}d)$.
    
    We start with giving a lower bound for $S_1$. Let $I_1$ be the set of all component indices $i$ in which $\F{y}_{u}(i)$ and $\F{y}_{v}(i)$ have opposite sign. Note that this implies that 
    \begin{align*}
        \abs{\F{y}_{u}(i) - (w_v/w_u)^{1/d}\F{y}_{v}(i)} = \abs{\F{y}_{u}(i)} + (w_v/w_u)^{1/d}\abs{\F{y}_{v}(i)}.
    \end{align*} Furthermore, note that we may express $\F{y}_{u}(i) = r_u \cdot \F{\tilde{z}}_{u}(i)/\norm{\F{\tilde{z}}_u}_p$. Since $w_u \le w_s\cdot c^d$ and $w_v \ge w_s$, we further have $(w_v/w_u)^{1/d} \ge 1/c$ and can thus rewrite $S_1$ as \begin{align*}
        S_1 &= \sum_{i \in I_1} \left( r_u \abs{ \frac{\F{\tilde{z}}_{u}(i)}{\|\F{\tilde{z}}_u\|_p} } + \left(\frac{w_v}{w_u} \right)^{1/d} r_v \abs{ \frac{\F{\tilde{z}}_{v}(i)}{\|\F{\tilde{z}}_v\|_p} } \right)^p \\
        &\ge \sum_{i \in I_1} \left( \left(r_u \frac{|\F{\tilde{z}}_{u}(i)|}{\|\F{\tilde{z}}_u\|_p}\right)^p + \left( \frac{r_v}{c} \frac{|\F{\tilde{z}}_{v}(i)|}{\|\F{\tilde{z}}_v\|_p} \right)^p \right)\\
        &=\frac{r_u^p}{\|\F{\tilde{z}}_u\|_p^p} \sum_{i \in I_1} |\F{\tilde{z}}_{u}(i)|^p + \frac{r_v^p}{c^p\|\F{\tilde{z}}_v\|_p^p} \sum_{i \in I_1} |\F{\tilde{z}}_{v}(i)|^p,
    \end{align*}
    where, in the second step, we used the inequality $(a + b)^p \ge a^p + b^p$ for all $a, b > 0$ and $p \ge 1$. 
    Now, we can apply tail bounds on the random variables in the above expression. We start with observing that the probability that $\tilde{z}_{ui}$, $\tilde{z}_{vi}$ have a opposite sign is exactly $1/2$. Hence, the set $I_1$ is a subset of component indices where each component is independently chosen with probability $1/2$. A Chernoff-Hoeffding bound (\Cref{thm:chernoff-hoeffding}) therefore implies that for every $\varepsilon > 0$, with overwhelming probability, \begin{align*}
        \frac{1}{2}d(1-\varepsilon) \le |I_1| \le \frac{1}{2}d(1+\varepsilon).
    \end{align*} 
    We further note that the random variables $||\F{\tilde{z}}_u||_p^p$, $||\F{\tilde{z}}_v||_p^p$, and $\sum_{i \in I_1} |\F{\tilde{z}}_{u}(i)|^p, \sum_{i \in I_1} |\F{\tilde{z}}_{v}(i)|^p$ are i.i.d random variables from $\rchi^p(d)$ and $\rchi^p(|I_1|)$, respectively. Hence, \Cref{cor:tailchisquared} and \Cref{cor:expectationchip}, imply that for every $\varepsilon > 0$, with overwhelming probability, \begin{align*}
        (1-\varepsilon) \frac{2d}{p} \le ||\F{\tilde{z}}_{u}||_p^p, ||\F{\tilde{z}}_{v}||_p^p \le (1+\varepsilon)\frac{2d}{p} 
   \end{align*} and \begin{align*} 
       (1-\varepsilon) \frac{2|I_1|}{p} \le \sum_{i \in I_1} |\F{\tilde{z}}_{u}(i)|^p, \sum_{i \in I_1} |\F{\tilde{z}}_{v}(i)|^p \le (1 + \varepsilon)\frac{2|I_1|}{p}. 
   \end{align*}
    Moreover, we note that the probability $\Pr{r_u \ge \zeta} = 1 - \zeta^{d}$ for every $0 < \zeta < 1$, so we have $r_u, r_v \ge \zeta$ with overwhelming probability. In total, this implies that with overwhelming probability, \begin{align*}
        S_1 &\ge \frac{\zeta^p}{(1+\varepsilon)2d/p}\frac{1}{2}(1-\varepsilon) \frac{2d}{p}(1-\varepsilon) + \frac{\zeta^p}{c^p(1+\varepsilon)2d/p}\frac{1}{2}(1-\varepsilon) \frac{2d}{p}(1-\varepsilon)\\
        &= \frac{\zeta^p (1-\varepsilon)^2 }{2(1+\varepsilon)} \left(1 + \frac{1}{c^p} \right).
    \end{align*}
    We note that by choosing $\zeta$ sufficiently large, and $c$ and $\varepsilon$ sufficiently small, we can push this lower bound to every number smaller than $1$. That is, we have shown that that for every $\varepsilon' > 0$, there are constants $\zeta < 1, c > 1$ such that with overwhelming probability, $S_1 \ge 1 - \varepsilon'$.

    We go on with lower bounding $S_2$. Analogously to $I_1$, let $I_2$ be the set of all component indices $i$ in which $y_{ui}$ and $y_{vi}$ have the same sign. This implies that
    \begin{align*}
        \abs{\F{y}_{u}(u) - (w_v/w_u)^{1/d}\F{y}_{v}(i)} = \abs{\abs{\F{y}_{u}(i)} - (w_v/w_u)^{1/d}\abs{\F{y}_{v}(i)} }.   
    \end{align*}
    We can hence reformulate $S_2$ as \begin{align*}
        S_2  &= \sum_{i \in I_2} \abs{ r_u\abs{\frac{\F{\tilde{z}}_{u}(i)}{\|\F{\tilde{z}}_u\|} } - \left( \frac{w_v}{w_u} \right)^{1/d} r_v\abs{\frac{\F{\tilde{z}}_{v}(i)}{\|\F{\tilde{z}}_v\|} } }^p\\
        &= \frac{r_u^p}{\norm{\F{\tilde{z}}_u}_p^p} \sum_{i \in I_2} \abs{ \abs{\F{\tilde{z}}_{u}(i)} - \left( \frac{w_v}{w_u} \right)^{1/d} \frac{r_v}{r_u} \frac{\norm{\F{\tilde{z}}_u}_p}{\norm{\F{\tilde{z}}_v}_p} \abs{\F{\tilde{z}}_{v}(i)} }^p.
    \end{align*}

    We first note that, since $|I_2| = d - |I_1|$ and with overwhelming probability $|I_1| = \Theta_d(d)$, we have $|I_2| = \Theta_d(d)$ with overwhelming probability. Furthermore, we have with overwhelming probability that $r_u, r_v \ge \zeta$ and that both $\norm{\F{\tilde{z}}_u}_p^p$ and $\norm{\F{\tilde{z}}_v}_p^p$ are between $(1-\varepsilon)2d/p$ and $(1+\varepsilon)2d/p$ just like in the above paragraph. Together with $(w_v/w_u)^{1/d} \le c$, this implies that with overwhelming probability, \begin{equation}\label{eq:smallertwo}
        \left( \frac{w_v}{w_u} \right)^{1/d} \frac{r_v}{r_u} \frac{\norm{\F{\tilde{z}}_u}_p}{\norm{\F{\tilde{z}}_v}_p} \le \frac{c}{\zeta} \left(\frac{1+\varepsilon}{1-\varepsilon}\right)^{\frac{1}{p}}.
    \end{equation}
    This bound can be made smaller than $2$ by choosing $c, \varepsilon$ small enough and $\zeta$ large enough. Furthermore, we get that for every $1 \le i \le d$ and any constant $\lambda > 0$, there is a constant probability of the event $\Event{\lambda}$ that $\abs{\F{\tilde{z}}_{u}(i)}$ is large enough and $\abs{\F{\tilde{z}}_{v}(i)}$ is small enough such that \begin{equation*}
        \abs{ \abs{\F{\tilde{z}}_{u}(i)} - 2\abs{\F{\tilde{z}}_{v}(i)}}^p \ge \lambda
    \end{equation*} because $\abs{\F{\tilde{z}}_{u}(i)}$ and $\abs{\F{\tilde{z}}_{v}(i)}$ are two independent samples from $\rchi_p(1)$\footnote{This is the crucial step of this proof in which our coupling between $\F{y}$ and $\F{\tilde{z}}$ turns out to be useful. This is because the components of $\F{\tilde{z}}$ are indeed independent whereas the components of $\F{y}$ are not.}. Hence, the sum \begin{equation}
        \sum_{i\in I_2} \abs{ \abs{\F{\tilde{z}}_{u}(i)} - 2\abs{\F{\tilde{z}}_{v}(i)}}^p
    \end{equation} is with overwhelming probability lower bounded by the sum of $|I_2| = \Theta_d(d)$ independent Bernoulli random variables with constant success probability. Therefore, a Chernoff-Hoeffding bound (\Cref{thm:chernoff-hoeffding}) implies that with overwhelming probability,
    \begin{equation}\label{eq:omegad}
        \sum_{i\in I_2} \abs{ \abs{\F{\tilde{z}_{u}}(i)} - 2\abs{\F{\tilde{z}_{v}}(i)}}^p = \Omega_d(d).
    \end{equation}
    As the bound from \Cref{eq:smallertwo} is with overwhelming probability smaller than $2$ for appropriate choices of $c, \varepsilon, \zeta$, we get that with overwhelming probability that \begin{equation*}
        \sum_{i \in I_2} \abs{ \abs{\F{\tilde{z}}_{u}(i)} - \left( \frac{w_v}{w_u} \right)^{1/d} \frac{r_v}{r_u} \frac{\norm{\F{\tilde{z}}_u}_p}{\norm{\F{\tilde{z}}_v}_p} \abs{\F{\tilde{z}}_{v}(i)} }^p = \Omega(d).
    \end{equation*} As we further get that $r_u^p/\| \F{\tilde{z}}_u \|_p^p  = \mathcal{O}_d(1/d)$, with overwhelming probability, we have in total that $S_2 = \Omega_d(1)$ with overwhelming probability where the leading constant does not depend on $c, \zeta, \varepsilon$.

    In total, we get that for every $\varepsilon' > 0$, with overwhelming probability, $S_1 + S_2 \ge 1 - \varepsilon' + \Omega_d(1)$ if we choose $c$ and $\varepsilon$ sufficiently small and $\zeta$ sufficiently large (i.e. sufficiently close to $1$). Hence, if we choose $\varepsilon'$ small enough, there is a $c>1$ such that with overwhelming probability, $S_1 + S_2 \ge c^p$. This implies our statement.
\end{proof} 
This lemma directly implies our main result.

\clusteringlp*
\begin{proof}
Similarly as in the proof of \Cref{thm:clusteringLinfty}, we use \Cref{lem:clusteringbound} and \Cref{lem:triangleclustering} to conclude that a.a.s.,\begin{align*}
    \textsc{CC}(G) = \Expected{\text{CC}(G_{\le n^{1/8}})} + o(1) \le \Pr{\Delta \mid \hat{v} \sim \hat{s}, \hat{t}} + o(1)
\end{align*} where $\hat{v},\hat{s},\hat{t}$ are three random vertices from $G_{\le n^{1/8}}$ conditioned on $w_{\hat{v}} \le w_{\hat{s}}, w_{\hat{t}}$. To bound, the above probability, fix any constant $c > 1$ and note that \begin{align*}
    \Pr{\Delta \mid \hat{v} \sim \hat{s}, \hat{t}} &\le \Pr{\Delta \mid (\hat{v} \sim \hat{s}, \hat{t}) \cap (\max\{w_{\hat{s}}, w_{\hat{t}}\} \le c^d w_{\hat{v}})} \\
    &\hspace{4cm} + \Pr{\max\{w_{\hat{s}}, w_{\hat{t}}\} \ge c^d w_{\hat{v}} \mid \hat{v} \sim \hat{s}, \hat{t} }.
\end{align*} Now we let $c$ be as in \Cref{lem:cluseringLpbound} and replace the first term by the bound derived in said lemma (\Cref{lem:cluseringLpbound}). The second term can be bounded by $(2 + o(1))c^{d(2-\beta)}$ which follows from the distribution of $\max\{w_{\hat{s}}, w_{\hat{t}}\}$ as derived in the proof of \Cref{thm:clusteringLinfty}. In total, this means that there are constants $a, b > 0, c > 1$ (from \Cref{lem:cluseringLpbound}) such that \begin{align*}
    \Pr{\Delta \mid \hat{v} \sim \hat{s}, \hat{t}} \le  a \exp(-b d) + (2 + o(1))c^{d(2-\beta)} = \exp(-\Omega_d(d))
\end{align*}
as desired. Note that the last step above holds since for sufficiently large $d$ there is a constant $\delta$ such that the above term is upper bounded by $\exp(-\delta d)$, which concludes the proof.
\end{proof}

\section{Testing for the dimensionality}

\label{sec:dimensiontest}

We turn to the question of how one can recover the dimension $d$ of $G(n, d, \beta, w_0)$ generated under $L_\infty$-norm. We show that this is accomplished by a simple test statistic that computes a modified version of the clustering coefficient in a subgraph of $G$ consisting of all vertices with weight in $[w_c, c  w_c]$ for some constants $w_c \ge w_0, 1 < c < 2/\sqrt{3}$. We show that the value of this test statistic is well concentrated around its expectation such that it allows us to test whether $G$ came from dimension $d$ as long as $d = o(\log(n))$. As a side result, this shows that all graphs in the low-dimensional regime form a clear dichotomy.

Formally, we consider the following statistical testing problem. We are given a graph $G$ on $n$ vertices, its weight sequence, and an integer $d = o(\log(n))$. Under the null hypothesis, $G$ is a GIRG generated in the weight sampling model with dimension $d$, whereas under the alternative hypothesis, $G$ was generated in dimension $d_1 \neq d$ or it is a Chung--Lu graph. Here, we allow $d_1$ to be any integer (potentially larger than $\log(n)$). 
As a proof of concept, we propose the following testing procedure for this problem. Fix a constant $1 < c < 2/\sqrt{3}$ and a weight $w_c \ge w_0$. Now, consider the induced subgraph $\Gtest$ of $G$ consisting of all vertices with weight in $[w_c, c w_c]$. For every vertex $v \in \Gtest$ that has at least two neighbors in $\Gtest$, we compute its local clustering coefficient $CC_{\Gtest}(v)$ and denote by $\textsc{CC}^{(+)}(\Gtest)$ the mean over all these values. We accept the null hypothesis if and only if condition (\ref{eq:testcondition}) is met. We show that the probability that this test makes a mistake under both the null and alternative hypothesis goes to zero as $n \rightarrow \infty$ and we capture this in the following statment.

\dimensionalitytesting*
\begin{proof}
    We start by estimating the expectation of $\textsc{CC}^{(+)}(\Gtest)$. It is not hard to see that by linearity of expectation \begin{align*}
        \Expected{|S| \hspace*{.1cm} \textsc{CC}^{(+)}(\Gtest)} &= \sum_{v \in G} \Expected{ \mathds{1}(v \in S) \hspace{.1cm} \textsc{CC}_{\Gtest}(v) }\\
        &= \sum_{v \in G} \Pr{v \in S} \Expected{ \textsc{CC}_{\Gtest}(v) \mid v \in S } \\
        &= \Expected{|S|}\Pr{\Delta \mid v\sim s,t}
    \end{align*} where $v,s,t$ are three random vertices in $\Gtest$, and $\Delta$ is the event that $v,s,t$ are a triangle. Notice that we used that $\Expected{ \textsc{CC}_{\Gtest}(v) \mid v \in S } = \Pr{\Delta \mid v\sim s,t}$ here as established (in the proof of) \Cref{lem:triangleclustering}; equality holds here because conditioning on $v \in S$ is the same as conditioning on $|\Gamma(v)| \ge 2$. 
    
    Our proof now proceeds in two steps: (1) we show that $\Pr{\Delta \mid v\sim s,t}$ is in the interval $(c^{-1} (3/4)^d, c(3/4)^d )$, and (2) we show that $\textsc{CC}^{(+)}(\Gtest)$ concentrates around its expectation using the method of typical bounded differences.

    For part (1), we apply \Cref{thm:sharpcliquebounds} and note that -- since $v,s,t$ are in $\Gtest$ -- the weights of $v,s,t$ differ by at most a factor of $c$. However, \Cref{thm:sharpcliquebounds} only yields a bound on the probability of $\Delta$ if $v$ is the vertex of minimal weight among $v,s,t$. However, if we define $\hat{v}, \hat{s}, \hat{t}$ to be the vertices $v,s,t$ reordered such that $\hat{v}$ is of minimal weight, we can express $\Pr{\Delta \mid v\sim s,t}$ as
    \[
    \frac{\Pr{\Delta}}{\Pr{v \sim s,t}} = \frac{\Pr{\hat{v} \sim \hat{s},\hat{t}}}{\Pr{v \sim s,t}} \Pr{\Delta \mid \hat{v} \sim \hat{s},\hat{t}} = \frac{w_{\hat{v}}}{w_v} \cdot \Pr{\Delta \mid \hat{v} \sim \hat{s},\hat{t}}
    \]
    where the last equality holds because $\Pr{v \sim s,t} = \lambda^2 w_v^2w_sw_t / n^2$ (we can ignore the minimum in \cref{eq:conprob} here because the weights are constant). Since the fraction in the above equation is at least $1/c$ and at most $1$, and since we can bound $\Pr{\Delta \mid \hat{v} \sim \hat{s},\hat{t}}$ using \Cref{thm:sharpcliquebounds}, we conclude that \begin{align*}
        \Pr{\Delta \mid v\sim s,t} \in \left( \frac{1}{c} \left( \frac{3}{4} \right)^d, c \left( \frac{3}{4} \right)^d\right)
    \end{align*}

    For the second part of the proof, we first show that $|S|$ is linear in $n$ with high probability. Afterwards, we apply the same procedure to $|S| \hspace{.1cm} \textsc{CC}^{(+)}(\Gtest)$.
    We start by showing that there is a constant $\alpha > 0$ such that $|S|$ is at least $\alpha n$ with probability $1 - n^{-\omega(1)}$. Consider a fixed vertex $v$ from $G$ and denote the number of its neighbors in $\Gtest$ by $X_v$. We note that every vertex in $G$ has a constant probability of being in $\Gtest$ and a probability of at least $\lambda w_0^2 / n$ to connect to $v$. $X_v$ is therefore lower bounded by the sum of $n$ independent Bernoulli random variables with success probability in $\Theta(1/n)$. Denote this sum by $\tilde{X}_v$ and note that $\mathbb{E}[\tilde{X}_v] = \Theta(1)$. By \cite[Proposition~1]{Cam_1960}, the $\tilde{X}_v$ thus converges to a Poisson distributed random variable with constant expectation. Accordingly, $\text{Pr}[\tilde{X}_v \ge 2]$ is constant as well. This shows that every vertex in $G$ has at least a constant probability of having two neighbors in $\Gtest$. As the probability that $v$ is in $\Gtest$ is constant as well, this implies that $\Expected{|S|} = \Omega(n)$. We continue with showing concentration of this random variable using \Cref{thm:bounded-diff}. We note that the random variables $x_1, x_2, \ldots, x_n$ (the positions of all vertices), and $w_1, w_2, \ldots, w_n$ (the weights of all vertices) are independent and define a product probability space $\Omega$ such that each $\omega \in \Omega$ defines a graph $G(\omega)$ and a corresponding test graph $\Gtest(\omega)$. We further define $f(\omega)$ as the value of $|S|$ in $G(\omega)$ and consider the ``bad'' event \begin{align*}
        \mathcal{B} = \{ \omega \in \Omega \mid  \text{the max degree in $\Gtest(\omega)$ is greater than $\log^3(n)$} \}.
    \end{align*}
    By \Cref{lem:strongdegreebound}, $\mathcal{B}$ happens with probability $n^{-\omega(1)}$ since all vertices in $\Gtest$ have at most constant weight. Now, let $\omega, \omega' \in \overline{\mathcal{B}}$ be such that they differ in at most two coordinates. Changing the weight or coordinate of one vertex can only decrease the number of vertices in $\Gtest$ with at least two neighbors by at most $2\log^3(n)$ as the weight or coordinate change only influences vertices that are neighbours of the changed vertex before or after the change. Accordingly, two coordinate or weight changes can only change $|S|$ by at most $c' \coloneq 4 \log^3(n)$. Using $t = n^{3/4}$ further satisfies the condition $t \ge 2M\Pr{\mathcal{B}}$ as $M \le n$ and $\Pr{\mathcal{B}} = n^{-\omega(1)}$. As $m = 2n$, we get,
    \begin{align*}
        \Pr{|S| - \Expected{|S|}| \ge n^{3/4}}\le & \hspace{.1cm} 2 \exp \left( - \frac{n^{1/2}}{32 \cdot 2 \cdot 16 \log^6(n) } \right) + \left( \frac{n^2}{\log^3(n)} + 1 \right)n^{-\omega(1)}\\
         =& \hspace{.1cm}  n^{-\omega(1)}.
    \end{align*}
    Similarly, we can show concentration of $f(S)=|S| \cdot \textsc{CC}^{(+)}(\Gtest) = \sum_{v \in S} \textsc{CC}_{\Gtest}(v)$. Again, changing the coordinate or weight of any two vertices can only increase or decrease the local clustering coefficient of at most $4\log^3(n)$ vertices by a value of at most one. Hence, we can again choose $c' \coloneq 4\log^3(n)$ and $t = n^{3/4}$ to obtain that
    \[
        \Pr{\left|f(S) - \Expected{f(S)}\right| \ge n^{3/4}} \le n^{-\omega(1)}.
    \]
    Combining these two concentration results, we get that \begin{align*}
        f(S) &= \Expected{|S|}\Pr{\Delta \mid v\sim s,t} \pm n^{3/4} \text{ and }\\
         |S| &= \Expected{|S|} \pm n^{3/4}
    \end{align*} both hold with probability $1 - o(1/n)$. Dividing by $|S|$ and using $\Expected{|S|} = \Theta(n)$ then yield that \begin{align*}
        \frac{f(S)}{|S|} = \textsc{CC}^{(+)}(\Gtest) &=\frac{\Expected{|S|}}{ \Expected{|S|} \mp n^{3/4} }\Pr{\Delta \mid v\sim s,t} \pm \frac{n^{3/4}}{{ \Expected{|S|} \mp n^{3/4} }}\\
        &= \Pr{\Delta \mid v\sim s,t} \pm n^{-1/5}.
    \end{align*} Using our estimate for $\Pr{\Delta \mid v\sim s,t}$ from part (1) of this proof concludes the argument.
\end{proof}

Using this, we immediately get that the probability that our test makes a mistake assuming that the null hypothesis is true is only $n^{-\omega(1)}$. Under the alternative hypothesis, assume that $d_1$ is the ground truth dimension $G$ came from, and assume further without loss of generality that $d_1 \ge d + 1$. We have to show that asymptotically, \begin{align*}
    \frac{1}{c} \left(  \frac{3}{4} \right)^d - n^{-1/5} &> c \left( \frac{3}{4} \right)^{d+1} + n^{-1/5}\\
    \Leftrightarrow 1 &> \frac{3}{4} c^2 + 2 \left( \frac{4}{3} \right)^d n^{-1/5}.
\end{align*} As $c < 2/\sqrt{3}$ and $d = o(\log(n))$, this inequality is true for sufficiently large $n$. To see this, observe that \begin{align*}
    \left( \frac{4}{3} \right)^d n^{-1/5} = \exp\left(\ln\left(\frac{4}{3} \right)d - \frac{1}{5}\ln(n)\right) = o(1).
\end{align*}

\acks{Andreas Göbel was funded by the project PAGES (project No. 467516565) of the German Research Foundation (DFG). We thank Marcos Kiwi for fruitful discussions related to this work.}

\bibliography{bibliography}

\appendix

\section{Experimental Data}

\begin{table}[h]
    \caption{The networks used in the plots in the first row of \Cref{fig:experiments} and their basic attributes. All found in the network repository \cite{nr} and the SNAP dataset \cite{snapnets}. \label{tab:graphs}}
     \begin{center}
       \begin{tabular}{ m{10em} m{5em} m{5em} m{10em} }
         \toprule
         Name             & $|V|$         & $|E|$         &  Category      \\
         \midrule
         ca-AstroPh       & \SI{18.7}{k}  & \SI{198.1}{k} &  collaboration \\
         ca-CondMat       & \SI{23.1}{k}  & \SI{93.4}{k}  &  collaboration \\
         ca-GrQc          & \SI{5.2}{k}   & \SI{14.5}{k}  &  collaboration \\
         ca-HepPh         & \SI{12}{k}    & \SI{118.5}{k} &  collaboration \\
         ca-MathSciNet    & \SI{332.7}{k} & \SI{820.6}{k} &  collaboration \\
         cit-patent       & \SI{3.7}{M}   & \SI{16.5}{M}  &  citation      \\
         cit-HepTh        & \SI{27.7}{k}  & \SI{352.8}{k} &  citation      \\
         cit-DBLP         & \SI{12.6}{k}  & \SI{49.7}{k}  &  citation      \\
         cit-HepPh        & \SI{34.5}{k}  & \SI{421.6}{k} &  citation      \\
         fb-pages-artists & \SI{50.5}{k}  & \SI{819.1}{k} &  social        \\
         soc-academia     & \SI{200.2}{k} & \SI{1.4}{M}   &  social        \\
         soc-youtube-snap & \SI{1.1}{M}   & \SI{3}{M}     &  social        \\
         socfb-A-anon     & \SI{3.1}{M}   & \SI{23.7}{M}  &  social        \\
         email-Enron      & \SI{36.7}{k}  & \SI{183.8}{k} &  social        \\
         bio-CE-CX        & \SI{15.2}{k}  & \SI{246}{k}   &  biological    \\
         bio-human-gene1  & \SI{21.9}{k}  & \SI{12.3}{M}  &  biological    \\
         bio-mouse-gene   & \SI{43.1}{k}  & \SI{14.5}{M}  &  biological    \\
         bio-WormNet-v3   & \SI{16.3}{k}  & \SI{762.8}{k} &  biological    \\
         bio-grid-human   & \SI{9.4}{k}   & \SI{62.4}{k}  &  biological    \\
         \bottomrule
       \end{tabular}
    \end{center}
    \end{table}

\end{document}